\documentclass[
 reprint,
 superscriptaddress,
 nofootinbib,
 amsmath,amssymb,
 aps,
 prx,
 onecolumn,
 tightenlines
]{revtex4-2}

\usepackage{graphicx}
\usepackage{dcolumn}
\usepackage{bm}
\usepackage[colorlinks,linkcolor=blue,urlcolor=blue, citecolor=blue]{hyperref}
\usepackage{bbm}
\usepackage[whole]{bxcjkjatype}
\usepackage[T1]{fontenc}

\usepackage{dsfont}
\usepackage{physics}
\usepackage{nicefrac}
\usepackage{enumitem}
\usepackage{tcolorbox}
\usepackage{needspace}
\usepackage{mathtools}
\usepackage{tabularx}
\newcolumntype{C}{>{\centering\arraybackslash}X}
\newcommand{\mc}{\mathcal}

\usepackage{amsmath, amssymb, amsfonts, amsthm}
\usepackage{thmtools, thm-restate}

\DeclareMathOperator{\Enc}{Enc}
\DeclareMathOperator{\Dec}{Dec}
\DeclareMathOperator{\EC}{EC}
\DeclareMathOperator{\id}{id}

\DeclareMathOperator{\Loc}{Loc}
\DeclareMathOperator{\eff}{eff}

\newtheorem{theorem}{Theorem}
\newtheorem*{theorem*}{Theorem}

\newtheorem{corollary}[theorem]{Corollary}

\newtheorem{definition}[theorem]{Definition}

\newtheorem{remark}[theorem]{Remark}

\begin{abstract}

Fault-tolerant capacities quantify the ability of a quantum channel to reliably transmit information when every component of the encoding and decoding procedure is noisy. Earlier work analyzed achievable communication rates under such noise using fault-tolerant implementations based on concatenated codes with a single logical qubit. In this work, we develop an alternative approach using concatenations of quantum Hamming codes, which offer constant space overhead by encoding many logical qubits simultaneously. We introduce modular techniques for implementing fault-tolerant circuits with quantum input/output interfaces using the concatenated quantum Hamming code.
These tools enable an analysis of fault-tolerant entanglement-assisted communication that is not only simpler, but also yields substantially higher achievable communication rates than previous methods, owing to the limited noise correlations in syndrome qubits of high-rate quantum Hamming codes.

\end{abstract}

\allowdisplaybreaks

\begin{document}
\author{Paula Belzig}
\affiliation{Institute for Quantum Computing and Department of Combinatorics and Optimization, \\ University of Waterloo, Waterloo, ON, Canada} 
\email{pbelzig@uwaterloo.ca}
\author{Hayata Yamasaki}
\affiliation{
Department of Computer Science, Graduate School of Information Science and Technology, The University of Tokyo, 7--3--1 Hongo, Bunkyo-ku, Tokyo, 113--8656, Japan
}
\email{hayata.yamasaki@gmail.com}

\title{Constant-space-overhead fault-tolerant quantum input/output and communication}

\maketitle

\tableofcontents

\section{Introduction}

\noindent For a wide range of quantum computing scenarios, it may be necessary at some point to transmit quantum information between distinct quantum systems. This is relevant for large quantum computations in distributed computing setups \cite{Y5,Wehnereaam9288,7562346}, gate teleportation, quantum repeaters, and, even more simply, for quantum communication setups where a sender and a receiver are connected by a noisy quantum channel and aim to reliably transmit an encoded message with vanishing error.

In order to communicate via a given quantum channel, using the tools of quantum Shannon theory, we would like to find and implement the best encoding and decoding procedure for the channel on quantum devices. However, quantum devices are highly susceptible to noise, and depending on the setup, this noise might be quite different from the channel noise.
Although noise effects in classical circuits are negligible on the time scales relevant for communication~\cite{Nicolaidis11}, the same assumption is not believed to be true for quantum circuits currently or in the future~\cite{Preskill18}.
Standard methods from fault-tolerant quantum computation~\cite{G,AGP06,yamasaki2022timeefficient} can reduce the effect of noise on a given circuit with \emph{classical} input and output for computation; however, in the context of communication, the sender's output and the receiver's input are quantum states that serve as inputs and outputs of the noisy quantum channel. In this work, we therefore develop tools for fault-tolerant quantum computation using the concatenated quantum Hamming code~\cite{yamasaki2022timeefficient} when the circuits have quantum input or output and show how these tools can be leveraged in order to show the achievability of non-zero communication rates even when the devices are noisy.

A fault-tolerant protocol simulates the original circuit with its error suppressed using the logical qubits of a quantum error-correcting code, even if the protocol is implemented by noisy physical qubits.
The trade-off between using more quantum resources and suppressing the logical error rate manifests in the protocol's \emph{space overhead}, which quantifies the required number of physical qubits per logical qubit, and the protocol's \emph{time overhead}, which quantifies the required time steps of physical gates per performing a logical gate. Previously,~\cite{CMH20,BCMH24} have shown that techniques using the concatenated 7-qubit Steane code~\cite{Steane_2002,KL96,AGP06} for protecting the encoding and decoding circuits of a communication protocol can lead to non-zero achievable rates for classical communication, quantum communication, and entanglement-assisted communication, even in the presence of independent and identically distributed (IID) Pauli noise in the quantum circuits for encoding and decoding.
Similar results were shown in~\cite{CFG24,tamiya2024polylogtimeconstantspaceoverheadfaulttolerantquantum} for more general classes of noise.
However, when using concatenated codes with a single logical qubit, protocols inevitably incur a polylogarithmically growing space overhead.

Progressing beyond concatenated many-to-one codes,~\cite{yamasaki2022timeefficient} introduced constructions of a quantum error-correcting code and a fault-tolerant protocol from concatenations of quantum Hamming codes, which share some properties with the 7-qubit Steane code but have many logical qubits and achieve an asymptotically constant space overhead.
However, the fault-tolerant protocol using the concatenated quantum Hamming code in~\cite{yamasaki2022timeefficient} is designed for quantum computation with classical inputs and outputs and thus cannot be applied directly to communication or distributed computing. Here, we close this gap by introducing and analyzing the notion of \emph{interfaced circuits} for this code, which provide fault-tolerant implementations of circuits using logical qubits of a quantum error-correcting code, equipped with interfaces for quantum inputs and outputs of states of physical qubits of the code.
We prove a level-conversion theorem and a threshold theorem for the constant-space-overhead interfaced circuits implemented using the concatenated quantum Hamming code, which enables the analysis of its logical action with a theoretical guarantee on the logical error rate.
We believe that these constructions and their resilience to errors are not only relevant for communication via a noisy quantum channel, but can also potentially find applications in other setups where a fault-tolerant circuit receives quantum input and output, such as third-generation quantum repeaters and magic state injection.

As an application of our methods, we study fault-tolerant entanglement-assisted communication, showing that the use of the concatenated quantum Hamming code simplifies and streamlines the analysis and also leads to substantially higher achievable communication rates compared to the previous work~\cite{BCMH24} based on the concatenated Steane code.
In particular, we derive lower bounds on the capacity for fault-tolerant entanglement-assisted communication that substantially narrow the gap to the upper bound and remain strictly positive for significantly higher, and thus more realistic, noise parameters (see Fig.~\ref{fig:capacity-comp1} for more details).

This manuscript is structured as follows. In Section~\ref{sec-prelim}, we introduce the relevant notions of capacity and the concatenated quantum Hamming code. In Section~\ref{sec-interfaced-circuits}, we introduce interfaced circuits and prove new level-conversion and threshold theorems for interfaced circuits, which allow us to deal with quantum inputs and outputs for fault-tolerant circuits implemented in the concatenated quantum Hamming code. In Section~\ref{sec-coding-thms}, we apply the level-conversion theorem to prove the coding theorems for the fault-tolerant entanglement-assisted classical capacity.

\section{Preliminaries}
\label{sec-prelim}

\subsection{Notation}

Let $\mathcal{M}_d$ denote the matrix algebra of $d\times d$ complex matrices. Quantum states of $d$-level quantum systems are given by positive semidefinite and unit-trace matrices $\rho\in\mathcal{M}_d$. We define quantum channels from a sender $A$ to a receiver $B$ as completely positive and trace-preserving (CPTP) maps $\mathcal{N}: \mathcal{M}_{d_A}\rightarrow \mathcal{M}_{d_B}$, where $d_A$ and $d_B$ is the dimensions of $A$ and $B$, respectively. Probability distributions $p$ of $d$ elements are represented as diagonal, positive semidefinite, and unit-trace matrices $\sum_{x}p(x)\ket{x}\bra{x}\in\mathcal{M}_{\mathrm{cl},d}\subset\mathcal{M}_{d}$, where $\mathcal{M}_{\mathrm{cl},d}$ is the algebra of diagonal matrices.
Quantum channels with classical input are defined as linear maps from $\mathcal{M}_{\mathrm{cl},d}$ that yield unit-trace positive semi-definite Hermitian matrices, and quantum channels with classical output are defined in the same way. The identity map on $\mathcal{M}_d$ is denoted by $\id_d$, and the identity map on classical bits is denoted by $\id_{\mathrm{cl},d}:\mathcal{M}_{\mathrm{cl},d}\to\mathcal{M}_{\mathrm{cl},d}$.

For measures of distance, we will use $F(\rho,\sigma)\coloneqq\qty(\Tr (\sqrt{\sqrt{\rho}\sigma \sqrt{\rho}}))^2$ to denote the fidelity of two quantum states $\rho$ and $\sigma$, and we will use
\begin{align}
    \| T-S\big\|_{1 \rightarrow 1} \coloneqq \sup \qty{\big\| (T-S)(\rho)\big\|_{\Tr} : \rho\in \mathcal{M}_{d_A}~\text{ is a quantum state}}
\end{align}
to denote the 1-to-1 distance of two quantum channels $T:\mathcal{M}_{d_A}\rightarrow \mathcal{M}_{d_B} $ and $S:\mathcal{M}_{d_A}\rightarrow \mathcal{M}_{d_B} $, where $\big\| \rho - \sigma\big\|_{\Tr} $ denotes the trace distance induced by the trace norm $\| \rho\|_{\Tr} \coloneqq \frac{1}{2}\| \rho\|_{1} =\frac{1}{2}  \Tr(\sqrt{ \rho ^{\dagger} \rho})$.

Quantum circuits form a dense subset of quantum channels written as a composition of a list of elementary gates, which we will specify in Section~\ref{sec-ft-setting}. We let $\Gamma$ denote a quantum circuit affected by noise, where a fault pattern $F$ (listing locations of the circuit and errors occurring at each location) occurs with a probability according to a distribution $\mathcal{F}$, and the faulty circuit is denoted by $[\Gamma]_{\mathcal{F}}$.
We will define the error model that is relevant to this work in Section~\ref{sec-ft-setting}.
A location in a circuit refers to each operation in the circuit.
Let $\Loc(\Gamma)$ denote the set of locations in the circuit $\Gamma$, and $|\Loc(\Gamma)|$ is the number of operations in $\Gamma$.

We write the binary entropy as $h_2(x)\coloneqq-x\log(x)-(1-x)\log(1-x)$ with $\log(x)$ referring to the logarithm with base $2$. We further write the von Neumann entropy of a quantum state $\rho\in\mathcal{M}_{d_A}$ as $H(A)_{\rho}\coloneqq-\Tr\left[\rho\log(\rho)\right]$.

A maximally entangled state of two qubits is denoted by $\phi_+=\ket{\phi_+}\bra{\phi_+}$, where $   \ket{\phi_+}\coloneqq\frac{1}{\sqrt{2}}(\ket{00}+\ket{11})$. 

\subsection{Fault-tolerant entanglement-assisted communication}

In classical communication, when a sender and a receiver want to communicate, their goal is to find a pair of procedures, an encoding procedure and a decoding procedure, so that copies of a noisy channel $\mathcal{N}$ transmit an encoded message with vanishing error and it can still be decoded reliably. The encoding and decoding procedures provide a \emph{coding scheme} for the channel $\mathcal{N}$, and the asymptotic rate of transmitted bits per channel use is an achievable communication rate for $\mathcal{N}$. The communication rate that is achieved for the best possible coding scheme for a given channel in the asymptotic limit is its capacity, which is an important quantifier of the channel's ability to transmit information and turns out to be equal to the mutual information~\cite{Shannon48} between the channel's input and output.

When the channel $\mc N$ connecting the sender and the receiver is quantum, different assumptions about the communication scenarios for a given quantum channel lead to different notions of capacity.
These notions include the classical capacity of a quantum channel~\cite{Holevo96,SW97}, which quantifies how well a quantum channel can transmit classical information encoded in a quantum state, and the quantum capacity~\cite{Lloyd97,Shor02,Devetak03}, which quantifies how well quantum information itself can be transmitted through the quantum channel. For these setups, \cite{CMH20} has constructed fault-tolerant coding schemes and shown lower bounds on the achievable rates when the coding scheme is subject to noise.
Entanglement-assisted communication setups lead to variants of these notions of capacity, which are often simpler to analyze; in this work, we focus on the entanglement-assisted classical capacity~\cite{BSST99,BSST02}, which is the classical capacity when the sender and the receiver share quantum entanglement, and the fault-tolerant entanglement-assisted capacity~\cite{BCMH24}.

In the entanglement-assisted communication setup, the sender and the receiver share quantum entanglement in the form of $nR_{\mathrm{ea}}$ copies of a maximally entangled state with some rate $R_{\mathrm{ea}}$ and aim to transmit a classical message. Then, the sender performs an encoding map $\mathcal{E}$ that encodes the $m$-bit classical message into a quantum state on $\mathcal{M}_{d_A}^{\otimes n}$. After the transmission through $n$ copies of the channel $\mathcal{N}$, the receiver receives a quantum state on $\mathcal{M}_{d_B}^{\otimes n}$ and applies a decoding map $\mathcal{D}$ with which they aim to recover the original message. The pair of maps $\mathcal{E}$ and $\mathcal{D}$ give a \emph{coding scheme} for the channel for entanglement-assisted communication.

\begin{definition}[Entanglement-assisted classical coding scheme]\label{def-ea-coding-scheme}
Let $\mathcal{N}:\mathcal{M}_{d_A} \rightarrow \mathcal{M}_{d_B}$ be a quantum channel, and let $n,m \in \mathbbm{N}$, $R_\mathrm{ea}\geq 0$, and $\epsilon>0$.
Then, an $(n,m,\epsilon,R_\mathrm{ea})$-coding scheme for entanglement-assisted classical communication consists of quantum channels $\mathcal{E}:\mathcal{M}_{\mathrm{cl},2}^{\otimes m}\otimes\mathcal{M}_2^{ \otimes \lfloor nR_\mathrm{ea} \rfloor}\rightarrow\mathcal{M}_{d_A}^{\otimes n}$ and $\mathcal{D}:\mathcal{M}_{d_A}^{\otimes n}\otimes \mathcal{M}_2^{ \otimes \lfloor nR_\mathrm{ea} \rfloor} \rightarrow\mathcal{M}_{\mathrm{cl},2}^{\otimes m}$ such that
\begin{align}
    F\Big(X,\mathcal{D} \circ  \big( (\mathcal{N}^{\otimes n } \circ \mathcal{E} )\otimes \id_2^{\otimes \lfloor nR_\mathrm{ea} \rfloor} \big) (X \otimes \phi_+^{\otimes \lfloor nR_\mathrm{ea} \rfloor })\Big) \geq 1-\epsilon
\end{align}
where $X=\sum_{x}p(x)\ket{x}\bra{x}$ for any probability distribution $p(x)$ over $m$-bit strings $x\in \{0,1\}^m$.
\end{definition}

\begin{definition}[Entanglement-assisted classical capacity]
Let $\mathcal{N}:\mathcal{M}_{d_A} \rightarrow \mathcal{M}_{d_B}$ be a quantum channel, and let $R_\mathrm{ea}\geq 0$ be the rate of entanglement-assistance.
For some $R_\mathrm{ea}$ and for every $n\in \mathbbm{N}$, consider an $(n,m(n),\epsilon(n),R_\mathrm{ea})$-coding scheme for entanglement-assisted classical communication.
A rate $R\geq 0$ is called achievable for entanglement-assisted classical communication via the quantum channel $\mathcal{N}$ if there exists $R_\mathrm{ea}\geq 0$ such that $R\leq \limsup_{n\rightarrow \infty} \Big\{ \frac{m(n)}{n} \Big\}$ and $\lim_{n\rightarrow \infty} \epsilon(n) \rightarrow 0$.
The entanglement-assisted classical capacity of $\mathcal{N}$ is given by
\begin{align}
  C^\mathrm{ea}(\mathcal{N}) = \sup \{ R : R \text{ achievable for entanglement-assisted communication via $\mathcal{N}$} \}.
\end{align}
\end{definition}

We define entanglement-assistance with respect to copies of the maximally entangled state $\phi_+$. This turns out to be sufficiently general as a consequence of~\cite{LP99,BDHSW14}, which show that arbitrary entangled states as a resource for communication do not increase the communication rates.

Importantly, the entanglement-assisted classical capacity is equal to the quantum mutual information of the channel input and output~\cite{BSST99,BSST02}:
    \[C^\mathrm{ea}(\mathcal{N}) = \sup_{\rho_{AA'}} I(B:A)_{(\id \otimes \mathcal{N})(\rho_{AA'})}\]
where the supremum is over all pure quantum states $\rho_{AA'}\in\mathcal{M}_{d_A}\otimes \mathcal{M}_{d_{A'}}$ and $I(B:A)_{\rho}=H(A)_{\rho}+H(B)_{\rho}-H(AB)_{\rho}$ is the quantum mutual information. In contrast to other notions of capacity, the entanglement-assisted classical capacity is special in the sense that it admits a characterization that is independent of the number of channel copies $n$, i.e., a single-letter characterization, and that its formula is a direct quantum analogue of the completely classical version from Shannon's noisy channel coding theorem~\cite[Section~13]{Shannon48}.

Given a fixed channel, if we want to use it for communication, it is important to find the best possible encoder and decoder in order to achieve the highest possible communication rate.
The encoder map and the decoder map then have to be implemented by circuits, and it is conventionally assumed that these circuits consist of noise-free gates and indeed implement the correct map. This would be justified for modern classical computers, where gate errors are negligible at the timescales relevant for communication protocols~\cite{Nicolaidis11}, but may not be true for quantum computers~\cite{Preskill18}. Due to this feature of noise in quantum computers, the same assumption as in the classical case of error-free gates, which has been made in most communication scenarios studied in quantum Shannon theory so far, may not be realistic in quantum communication. In light of this,~\cite{CMH20} initiated the study of communication under an error model describing noise in the encoding and decoding operations, introducing a notion of capacity for fault-affected encoder and decoder circuits for classical and quantum capacity of a quantum channel. Subsequently,~\cite{BCMH24} extended~\cite{CMH20} to the case of entanglement-assisted classical capacity. 

Here, we summarize the definitions of coding schemes, achievable rates, and capacities under noise assumptions on the encoding and decoding maps as introduced in \cite{CMH20,BCMH24,CFG24}.
To simplify the presentation, we consider a scheme where the entangled resource for communication is available to the sender and the receiver in the code space.
In an alternative model, we may consider perfect maximally entangled physical states that must be encoded into the code spaces via noisy gates and would subsequently be transformed into maximally entangled states in the code space using entanglement distillation 
~\cite{DW03}; in such a model, by fault-tolerantly implementing the scheme from~\cite{DW03} as a subroutine of the coding scheme using the concatenated quantum Hamming code (similarly to \cite{BCMH24}), we see no obstacle to obtaining versions of our main results for this alternative assumption.
For the circuits $\mathcal{E}$ of the encoder map and $\mathcal{D}$ of the decoder map, under an error model $\mathcal{F}(p_0)$ parameterized by physical error rate $p_0$, we write the corresponding faulty circuits as $[\mathcal{E}]_{\mathcal{F}(p_0)}$ and $[\mathcal{D}]_{\mathcal{F}(p_0)}$, respectively.

\begin{definition}[Fault-tolerant entanglement-assisted classical coding scheme]\label{defn:FTEACS}
Let $\mathcal{N}:\mathcal{M}_{d_A} \rightarrow \mathcal{M}_{d_B}$ be a quantum channel, and let $n,m \in \mathbbm{N}$, $R_\mathrm{ea}\geq 0$ and $\epsilon>0$.
For $0\leq p_0 \leq 1$, let $\mathcal{F}(p_0)$ denote an error model.
Then, an $(n,m,\epsilon,R_\mathrm{ea})$-coding scheme for fault-tolerant entanglement-assisted communication consists of quantum circuits $\mathcal{E}:\mathcal{M}_{\mathrm{cl},2}^{\otimes m}\otimes\mathcal{M}_2^{ \otimes  \lfloor nR_\mathrm{ea}\rfloor }\rightarrow\mathcal{M}_{d_A}^{\otimes n}$ and $\mathcal{D}:\mathcal{M}_{d_A}^{\otimes n}\otimes \mathcal{M}_2^{ \otimes \lfloor nR_\mathrm{ea}\rfloor } \rightarrow\mathcal{M}_{\mathrm{cl},2}^{\otimes m}$ such that
\begin{align}\label{eq-coding-error-fidelity}
F\big(X,\big[\mathcal{D}\big]_{\mathcal{F}(p_0)} \circ  \big( (\mathcal{N}^{\otimes n} \circ \big[\mathcal{E}\big]_{\mathcal{F}(p_0)}\big) \otimes \id_2^{\otimes  \lfloor nR_\mathrm{ea}\rfloor }  \big)   (X \otimes \phi_+^{\otimes  \lfloor nR_\mathrm{ea}\rfloor })\big) \geq 1-\epsilon,
\end{align}
where $X=\sum_{x}p(x)\ket{x}\bra{x}$ for any probability distribution $p(x)$ over $m$-bit strings $x\in \{0,1\}^m$.
\end{definition}

\begin{definition}[Fault-tolerant entanglement-assisted classical capacity]
Let $\mathcal{N}:\mathcal{M}_{d_A} \rightarrow \mathcal{M}_{d_B}$ be a quantum channel, and let $R_\mathrm{ea}\geq 0$ be the rate of entanglement-assistance.
For $0\leq p_0 \leq 1$, let $\mathcal{F}(p_0)$ denote an error model.
For some $R_\mathrm{ea}$ and for every $n\in \mathbbm{N}$, consider an $(n,m(n),\epsilon(n),R_\mathrm{ea})$-coding scheme for fault-tolerant entanglement-assisted classical communication under the noise model $\mathcal{F}(p_0)$.
A rate $R\geq 0$ is called achievable for fault-tolerant entanglement-assisted classical communication via the quantum channel $\mathcal{N}$ if there exists $R_\mathrm{ea}\geq 0$ such that $R\leq \limsup_{n\rightarrow \infty} \Big\{ \frac{m(n)}{n} \Big\}$ and $\lim_{n\rightarrow \infty} \epsilon(n) \rightarrow 0$.

The fault-tolerant entanglement-assisted classical capacity of $\mathcal{N}$ is given by
  \begin{equation*}
\begin{split}
  C^\mathrm{ea}_{\mathcal{F}(p_0)}(\mathcal{N}) =\sup &\{R : R \text{ achievable rate for fault-tolerant entanglement-assisted communication via $\mathcal{N}$}\}.
  \end{split}
\end{equation*}
\end{definition}

In principle, any quantum circuits $\mathcal{E}$ and $\mathcal{D}$ may be chosen in Definition~\ref{defn:FTEACS} to define a coding scheme for fault-tolerant entanglement-assisted classical communication. However, only one very specific coding scheme design based on the concatenated 7-qubit Steane code is currently proven to have achievable rates where the noiseless case is recovered as the physical error rate $p_0$ vanishes.
Under the assumptions listed in this section, the following lower bound on the fault-tolerant entanglement-assisted classical capacity of a channel $\mathcal{N}:\mathcal{M}_{d_A}\rightarrow \mathcal{M}_{d_B}$ with $j_1=\log(d_A)$, $j_2=\log(d_B)$ was proven in \cite{BCMH24}:
\begin{equation} \label{eq-bcmh-bound}
    C^\mathrm{ea}_{\mathcal{F}(p_0)}(\mathcal{N})\geq C^\mathrm{ea}(\mathcal{N})-f_\mathrm{CONT}(p_0)-f_\mathrm{AVP}(p_0)\end{equation}
with $f_\mathrm{CONT}(p_0)=8c(j_1+j_2)p_0j_2+(1+4c(j_1+j_2)p_0)h\qty(\frac{4c(j_1+j_2)p_0}{1+4c(j_1+j_2)p_0})$ and $f_\mathrm{AVP}(p_0)=2\sqrt{2{j_2p_0}} \big(2^{j_1+j_2} (j_1+j_2) +1\big) \left|2\log( \frac{2(j_1+j_2)cp_0 }{2^{j_1j_2}})\right|  +2h\qty(\sqrt{2{j_2p_0}} 2^{j_1+j_2} \left|2\log( \frac{2(j_1+j_2)cp_0 }{2^{j_1j_2}})\right| )+ 2 (j_1+j_2)cp_0j_2 =\mathcal{O}(p_0\log p_0)$.

In this work, we will develop new methods for fault-tolerant communication that substantially improve upon this previous lower bound. 
To achieve this, we will develop a construction for coding schemes based on the concatenated quantum Hamming codes introduced in~\cite{yamasaki2022timeefficient}, which we review in the next section.

\subsection{Concatenated quantum Hamming codes}

\label{sec-ft-setting}

In order to design a coding scheme for fault-tolerant communication, we will be using methods from the conventional setting of fault-tolerant quantum computation in order to compile the original circuits for encoding and decoding information into a circuit on physical qubits to simulate the original circuit at the logical level of quantum error-correcting codes. We call this compiled circuit on physical qubits a \textit{fault-tolerant circuit}. The fault-tolerant circuit is composed of a set of physical operations, i.e., preparations, gates, measurements, and wait,  i.e., performing the identity gate $\mathds{1}$.
Our model is as follows:

\begin{enumerate}
  \item \textbf{Circuit-level local stochastic Pauli error model}: Following the previous works~\cite{CMH20,BCMH24}, we assume that the fault-tolerant circuit undergoes a conventional error model, a \textit{local stochastic Pauli error model}. We say that a circuit undergoes a local stochastic Pauli error model if the faults occurring in the circuit satisfy the following: (i) a set $F$ of faulty locations is stochastically specified in such a way that any set $S$ of locations satisfies $\Pr[F\supset S]\leq\prod_{j\in S}p_{0,j}$, where $p_{0,j}\geq 0$ is a parameter associated to each location $j$ of the circuit; (ii) the operation at each faulty location in $F$ is replaced with a multiqubit Pauli channel $\mathcal{S}_j$ that may depend on the location.
  We may write
  \begin{equation}
      p_0\coloneqq\max_{j}\{p_{0,j}\},
  \end{equation}
  where the maximum is taken over all locations $j$ in the circuit.
  We call $p_{0,j}$ the physical error rate of the location $j$ and, for simplicity, may also call $p_0$ the physical error rate if obvious from the context. We call $p_{0,j}$ the physical error rate of the location $j$. We will usually refer to this error model as the \emph{local stochastic Pauli error model with noise parameter $p_0$} and, for simplicity, may also call $p_0$ the physical error rate if obvious from the context.
  \item \textbf{Parallel quantum operation}: The operations on physical qubits can be performed in parallel, where each qubit is involved in at most one operation at a time.
  \item \textbf{Faultless classical computation}: Conditioned on the outcome of measurements on physical qubits, we can perform classical computation to change the subsequent operations on the qubits adaptively. For simplicity of analysis, we assume that classical computation is faultless.
  \item \textbf{Parallel classical computation but nonzero runtime}: Similar to the operations on physical qubits, classical computation is assumed to be performed in parallel so that the depth of classical circuits determines the runtime of classical computation, as in~\cite{yamasaki2022timeefficient}.
    In our setting, errors may also occur on these wait operations according to the local stochastic Pauli error model.
  \item \textbf{Allocation of qubits and bits}: We allocate qubits for the preparation and deallocate them for the measurement. The measurement allocates bits for classical computation to process the measurement outcome. After using the measurement outcomes, we discard the outcomes and deallocate the bits. The number of qubits and bits at a time is that which has already been allocated before and is not yet deallocated at the time.
  \item \textbf{No geometrical constraint}: We assume that two-qubit gates are applicable to any pair of physical qubits without geometrical constraint as in photonic systems~\cite{yamasaki2020polylogoverhead,PhysRevResearch.2.023270,PhysRevLett.123.200502,PhysRevLett.112.120504,doi:10.1063/1.5100160} and distributed architectures~\cite{Y5,Wehnereaam9288,7562346}.
\end{enumerate}

Note that~\cite{G,yamasaki2022timeefficient} give techniques for implementing circuits for the local stochastic error model, which our work follows; by contrast,~\cite{CMH20,BCMH24} considered the IID Pauli error model, which is a special case of the general local stochastic error model.

We now summarize the construction of concatenated codes $\mathcal{Q}^{(L)}$ obtained by concatenating quantum Hamming codes as proposed in~\cite{yamasaki2022timeefficient}.
With $N_r\coloneqq 2^r-1$, let $\mathcal{C}_r$ ($r=2,3,\ldots$) denote the family of $[N_r,N_r-r,3]$ classical Hamming codes~\cite{6772729}, which have $N_r$-bit block length, $(N_r-r)$-bit dimension, and distance $3$~\cite{10.5555/552386}.
Quantum Hamming codes $\mathcal{Q}_{\text{H},r}$ ($r=3,4,\ldots$)~\cite{PhysRevA.54.4741} are Calderbank-Shor-Steane (CSS) codes~\cite{N4,PhysRevA.54.1098,S3} of $\mathcal{C}_r$ over its dual code $\mathcal{C}_r^\perp$, which are in a family of $[[N_r,K_r,3]]$ codes having $N_r=2^r-1$ physical qubits, $K_r\coloneqq N_r-2r$ logical qubits, and distance $3$.
We use $\mathcal{Q}_{l}=\mathcal{Q}_{\text{H},r_l}$ with parameter $r_l\coloneqq l+2$ for the concatenation at each level $l\in\{1,\ldots,L\}$, which leads to the sequence $\mathcal{Q}_{\text{H},3},\mathcal{Q}_{\text{H},4},\ldots$ of quantum Hamming codes starting from Steane's $7$-qubit code $\mathcal{Q}_{\text{H},3}$~\cite{PhysRevLett.77.793,S3} at level $1$, a $[[15,7,3]] $ code at level $2$, and a $[[31,21,3]]$ code at level $3$. For simplicity of notation, we will from now on write $K_l=K_{r_l}$ and $N_l=N_{r_l}$. Overall, the rate of logical qubits and physical qubits converges to $\frac{K_{l}}{N_{l}}\to 1$ as $l\to\infty$.

The concatenated code $\mathcal{Q}^{(L)}$ is constructed by iteratively defining $\mathcal{Q}^{(l)}$ for $l=L,L-1,\ldots,1$.
Let $K^{(l)}$ and $N^{(l)}$ denote the numbers of logical and physical qubits of $\mathcal{Q}^{(l)}$, respectively, which turn out to be $K^{(l)}=\prod_{l^{\prime}=1}^{l}K_{{l^{\prime}}}$, $N^{(l)}=\prod_{l^{\prime}=1}^{l}N_{{l^{\prime}}}$.
We define a \textit{level-$l$ register} as a collection of $K^{(l)}$ logical qubits of $\mathcal{Q}^{(l)}$,
where a physical qubit is referred to as a \textit{level-$0$ register} (or level-$0$ qubit). This collection of $K^{(l)}$ logical qubits will then serve as a unit that is treated as a "level-$l$ qubit" for the next level.

To define the concatenated code, we present the relation between $\mathcal{Q}^{(l)}$ and $\mathcal{Q}^{(l-1)}$ for each $l$.  To be precise, we will be using $N_l$ copies of $Q^{(l-1)}$ as building blocks of $\mc Q^{(l)}$.
To begin, we divide the $K^{(l)}=K^{(l-1)}\times K_{l}$ logical qubits in the level-$l$ register for $\mathcal{Q}^{(l)}$ into $K^{(l-1)}$ blocks of $K_{l}$ qubits.
Then, for each $k^{(l-1)}\in\{1,\ldots, K^{(l-1)}\}$, picking the $k^{(l-1)}$th qubit from each of the $N_{l}$ level-$(l-1)$ registers for $\mathcal{Q}^{(l-1)}$, we encode the $K_{l}$ qubits in the $k^{(l-1)}$th block into the picked $N_{l}$ qubits as the logical qubits of the quantum code $\mathcal{Q}_{l}$.

We design this concatenated code so that even if all the qubits in one of the level-$(l-1)$ registers suffer from correlated errors, we can still recover the encoded level-$l$ register. In this way, we obtain an $[[N^{(L)},K^{(L)},D^{(L)}]]$ concatenated code $\mathcal{Q}^{(L)}$ at the concatenation level $L$ with parameters
\begin{align}
    N^{(L)}=\prod_{l^\prime=1}^{L}N_{l^\prime},\\
    K^{(L)}=\prod_{l^\prime=1}^{L}K_{l^\prime},\\
    D^{(L)}=\prod_{l^\prime=1}^{L}D_{l^\prime}.
\end{align}
This construction of a concatenated code has a non-vanishing overhead rate~\cite{yamasaki2022timeefficient} of
\begin{equation}
    \frac{N^{(L)}}{K^{(L)}}=O(1)\quad\text{as $L\to\infty$}.
\end{equation}

As a concrete example, this leads to a $[[3255,147,27]]$ code $\mc Q^{(3)}$, a $[[105,7,9]]$ code $\mc Q^{(2)}$, and recovers the $7$-qubit Steane code $Q^{(1)}=Q_1=Q_{H,3}$ for one level of concatenation.\footnote{In comparison, the concatenated $7$-qubit Steane code at level $3$ is a $[[343,1,3]]$ code; this code has a rate of $343$ physical qubits per logical qubit, whereas the concatenated Hamming code at level $3$ merely needs $\frac{3255}{147}\approx 22.14$ physical qubits per logical qubit.} In order to construct $Q^{(3)}$, associate the $K^{(3)}=147$ logical qubits (which form a level-$3$ register) with $K^{(2)}=7$ blocks of $K_3=21$ qubits. Then, the $K_3=21$ logical qubits (of $Q^{(3)}$) from the first block are associated with $N_3=31$ qubits by using the Hamming code $Q_3$, and the same is repeated for the other $6$ blocks. Importantly, however, this encoding of $21$ qubits into $31$ qubits is performed in a specific way in order to counteract the correlated error that would otherwise arise. Specifically, we take $31$ copies of $\mc Q^{(2)}$, noting that each of them encodes $7$ logical qubits of $\mc Q^{(2)}$ (one level-$2$ register) into $105$ physical qubits, from which we obtain a total of $3255=31\times105$ physical qubits. Now, the $K_3=21$ logical qubits (of $Q^{(3)}$) in each block are encoded into $31$ physical qubits where each physical qubit originates from a different copy of $\mc Q^{(2)}$. This structure, where the encoding based on $Q_l$ is distributed across copies of $Q^{(l-1)}$, is integral for ensuring that potential correlated errors at one level do not spread in a correlated way to next level.

This fault-tolerant protocol implementing a circuit with $\mathcal{Q}^{(L)}$ suppresses the logical error rate in the following way:

\begin{theorem}[{Threshold theorem, \cite[Section~E]{yamasaki2022timeefficient}}]
For any $L\in\mathbbm{N}$,
let $\mathcal{Q}^{(L)}$ be the $[[N^{(L)},K^{(L)},D^{(L)}]]$ concatenated quantum Hamming code with concatenation level $L$ and threshold $p_{\textrm{th}}$. Let $U:\mathcal{M}_2^{\otimes K^{(L)}}\to\mathcal{M}_2^{\otimes K^{(L)}}$ be a $K^{(L)}$-qubit quantum circuit, and let $U_L:\mathcal{M}_2^{\otimes N^{(L)}}\to\mathcal{M}_2^{\otimes N^{(L)}}$ denote its implementation in $\mathcal{Q}^{(L)}$ (which is a $N^{(L)}$-qubit quantum circuit). Suppose that $\mathcal{F}(p_0)$ denotes the local stochastic Pauli error model with noise parameter $p_0$.
Then, there exists a threshold $p_\mathrm{th}>0$ such that for any $L$ and for any $0\leq p_0< p_\mathrm{th}$, we have
\[\big\| U-[U_{L}]_{\mathcal{F}(p_0)}\big\|_{1\rightarrow 1} \leq C \qty(\frac{p_0}{p_\mathrm{th}})^{2^L}p_\mathrm{th}  |\Loc(U)|\]
with some constant $C>0$.
 \end{theorem}

\subsection{Main results}
\label{sec-strategy}

Standard techniques from quantum fault-tolerance often cannot be directly applied to the problem of communication or only allow for substantially reduced communication rates. Strategies with one (large) fault-tolerant implementation, where the communication channel is considered part of the circuit noise, will only enable communication for very specific channels: channels that are very close to the identity, i.e. with noise below the threshold, and channels for which the code space of the error-correcting code is well suited (see also Remark~\ref{remark-multiplicative-bound}). In general, however, we are interested in setups where the noise affecting the communication channel may differ significantly from the noise affecting the local computation. In particular, the noise affecting the \emph{channel} does not have to be below the fault-tolerance threshold in the context of our results.

To obtain fault-tolerant communication rates for such quantum channels, we will present a construction of a coding scheme with a fault-tolerant implementation in the concatenated quantum Hamming code in terms of interfaced circuits. Interfaced circuits, which we will introduce in Section~\ref{sec-interfaced-circuits}, implement the original circuit using the quantum error-correcting code together with additional circuits that serve as the interface between information in the code space and the physical system. This implementation is similar to the construction in~\cite{CMH20,BCMH24} for the concatenated 7-qubit Steane code; however, there are some notable differences.
Firstly, we will construct the interfaces as gadgets of the code at each concatenation level, which leads to a more direct compilation and analysis of the fault-tolerant protocol; secondly, the concatenated quantum Hamming code asymptotically achieves constant space-overhead by encoding a larger number of physical qubits, leading to a block-wise structure on the physical qubit systems.

As has been noted for some examples of quantum error-correcting codes in~\cite{gottesman2014faulttolerant}, \cite[Lemma~11]{8555154},~\cite{MGHHLPP14}, \cite[Theorem~III.3]{CMH20} and \cite{CFG24}, for circuits with quantum input and output of physical qubits, techniques from fault-tolerance do not guarantee that the logical error can be suppressed arbitrarily. This also manifests for the concatenated quantum Hamming code, and we show here that (similar to previous results for other codes) the logical error on a single-qubit input and/or output is proportional to the single gate error in our level-conversion theorem for the concatenated quantum Hamming code (Corollary~\ref{thm-threshold-interfaced}).

Using interfaced circuits implemented in the concatenated quantum Hamming code $\mathcal{Q}^{(L)}$, we construct a fault-tolerant entanglement-assisted coding scheme for a given channel $\mathcal{N}$ based on a coding scheme for communication via an effective channel with correlated channel uses with \emph{no} errors on the encoder and decoder. We thereby demonstrate that the strategy employed in~\cite{CMH20,BCMH24} is not limited to the concatenated 7-qubit Steane code, but rather generalizes to other concatenated codes. In addition, the constant space-overhead of the fault-tolerant protocol we consider limits the amount of correlation between the channel uses due to the small number of syndrome qubits. Owing to this limited correlation, the effective scenario in our case of using the constant-space-overhead protocol is simpler than the one obtained in~\cite{CMH20,BCMH24}. For this reason, our construction of a coding scheme in Theorem \ref{thm-cap-lower-bound} builds more directly on previous results, allows for a substantially tighter version of an important bound from~\cite{CMH20,BCMH24}, and thereby achieves higher rates for fault-tolerant entanglement-assisted communication with this code.

\section{Interfaced circuits for concatenated quantum Hamming codes}
\label{sec-interfaced-circuits}

In this section, we introduce and analyze a fault-tolerant protocol using concatenated quantum Hamming codes to implement quantum circuits with quantum input and output.
In Section~\ref{sec-definition-interfaced-circuits}, we introduce the notion of an interfaced circuit with quantum input and output. In Section~\ref{sec-compile}, we present our protocol that compiles a given circuit with quantum input and output into a circuit implemented by physical qubits.
Then, in Section~\ref{sec-gadgets}, we clarify the explicit construction of the relevant gadgets and give conditions for the fault tolerance of gadgets. In Section~\ref{sec-correctness-conditions} and \ref{sec-level-conversion}, we analyze the effective error rates in the implemented circuits and establish the level-conversion theorem for interfaced circuits.

\subsection{Definition: Interfaced circuit}
\label{sec-definition-interfaced-circuits}

In our protocol, the qubits of the original circuit are represented by the qubits in level-$L$ registers of the concatenated quantum Hamming code $\mathcal{Q}^{(L)}$ at a concatenation level $L$, which leads to the level-$L$ circuit composed of level-$L$ operations acting on level-$L$ registers.
The required concatenation level $L$ will be selected according to our desired logical error rate.
Then, we will derive the physical implementation of the intended level-$L$ circuit with a level-$0$ circuit by recursively compiling the level-$l$ circuit into the corresponding level-$(l-1)$ circuit for $l=L, L-1, \ldots, 1$.

However, the protocol here is different from that of~\cite{yamasaki2022timeefficient} in that~\cite{yamasaki2022timeefficient} only considers computation with classical inputs and outputs; by contrast, the compiled circuits here are intended to be able to conduct communication tasks through physical quantum channels and thus need to have interfaces between the logical and physical levels for quantum inputs and outputs.
Thus, we introduce an extended definition of level-$l$ circuits to include interfaces to lower concatenation levels, which we call \textit{interfaced circuits}. Then, we present the protocol for compiling such interfaced circuits.

Given an original circuit $U:\mathcal{M}_2^{\otimes K^{(L)}} \rightarrow \mathcal{M}_2^{\otimes K^{(L)}}$ representing a $K^{(L)}$-qubit unitary transformation, our protocol will provide an implementation containing a circuit $U_{Q^{(L)}}:\mathcal{M}_2^{\otimes N^{(L)}}\rightarrow \mathcal{M}_2^{\otimes N^{(L)}}$ on $N^{(L)}$ physical qubits which represents the conventional implementation of $U$ at level $L$ of $\mc Q^{(L)}$, as well as interfaces to unencoded qubits.

We write the implementation of the $K^{(L)}$-qubit unitary $U$ in $Q^{(l)}$ at the levels $l=0,1,...,L-1,L$ as follows:
\begin{align}
    U_{Q^{(0)}}&:\mathcal{M}_2^{\otimes K^{(L)}}\rightarrow \mathcal{M}_2^{\otimes K^{(L)}}
    \\
    U_{Q^{(1)}}&:\mathcal{M}_2^{\otimes N^{(L-1)}K_L}\rightarrow \mathcal{M}_2^{\otimes N^{(L-1)}K_L}\\ &\vdots\notag\\
    U_{Q^{(L)}}&:\mathcal{M}_2^{\otimes N^{(L)}}\rightarrow :\mathcal{M}_2^{\otimes N^{(L)}}
\end{align}
All of the above circuits are equivalent to $U$, but acting on registers rather than qubits. Correspondingly, the implementation of $U$ at level $L$ can be understood as both of the following objects:
\begin{align}
    U_{Q^{(L)}}&:\qty(\text{level-$L$ register})\rightarrow :\qty(\text{level-$L$ register})
    \\
    U_{Q^{(L)}}&:\qty(\text{level-$0$ register})^{\otimes N^{(L)}}\rightarrow :\qty(\text{level-$0$ register})^{\otimes N^{(L)}}
\end{align}
although its composition in terms of level-$l$ gates might look quite different. The conversion between these equivalent objects is outlined in the compilation procedure that we present later. In the circuit diagrams that follow, the input lines depict level-$l$ registers rather than physical qubits, and the gates depict level-$l$ gates rather than their physical implementation.

In the following, we define the objects $\Dec^{(L)}:\mathcal{M}_2^{\otimes N^{(L)}}\rightarrow \mathcal{M}_2^{\otimes K^{(L)}}$ and $\Enc^{(L)}:\mathcal{M}_2^{\otimes K^{(L)}}\rightarrow \mathcal{M}_2^{\otimes N^{(L)}}$ which map between level $0$ and level $L$ of $\mc Q^{(L)}$, and refer to the following object
\begin{align}
    \Dec^{(L)} \circ U_{Q^{(L)}} \circ \Enc^{(L)} &: \mathcal{M}_2^{\otimes K^{(L)}}\rightarrow \mathcal{M}_2^{\otimes K^{(L)}},
\end{align}
as an interfaced circuit.

More precisely, we construct $\Dec^{(L)}$ and $\Enc^{(L)}$ in a recursive fashion from intermediate interfaces between levels $l$ and $l-1$ for $l=L,L-1,...,2,1$. This construction is similar to how $\mc Q^{(L)}$ is constructed; the intermediate encoding interface corresponds to a circuit which performs the step where the qubits are encoded from level $l-1$ to $l$ according to the quantum Hamming code $\mc Q_l$. In the course of this, we use the same mechanism to redistribute the qubits to control the spread of errors.

For each $l\in\{1,\ldots,L\}$, we define a level-$l$ intermediate encoding interface $\Enc_l$ and a level-$l$ intermediate decoding interface $\Dec_l$. The level-$l$ intermediate encoding interface takes an input of $K_l$ copies of a level-$(l-1)$ register of $Q^{(l-1)}$, and outputs a level-$l$ register. Conceptually, the encoding interface thus maps between registers; as a circuit, $\Enc_l$ takes an input of $K_l$ copies of $K^{(l-1)}$ logical qubits encoded in $N^{(l-1)}$ physical qubits (= a total of $K_lN^{(l-1)}$ physical qubits) to an output of $N_l$ copies of a level-$(l-1)$ register of $Q^{(l-1)}$ by encoding $K^{(l)}$ logical qubits in $N^{(l)}$ physical qubits according to the construction of the concatenated code $\mathcal{Q}^{(l)}$ (= a total of $N^{(l)}$ physical qubits). In other words, we define the intermediate encoding interface as a circuit
\begin{align}
\Enc_l:\mathcal{M}_2^{\otimes K_l N^{(l-1)}} \rightarrow \mathcal{M}_2^{\otimes N_l N^{(l-1)}}  \end{align}
which, in the language of registers, is equivalent to 
\begin{align}
    &\Enc_l:\qty(\text{level-$(l-1)$ register})^{\otimes K_l}\to\qty(\text{level-$l$ register})^{\otimes 1}.
\end{align}

The level-$l$ intermediate decoding interface $\Dec_l$, on  the other hand, takes an input of a level-$l$ register and outputs $K_l$ copies of level-$(l-1)$ registers, defined as
\begin{align}\Dec_l:\mathcal{M}_2^{\otimes N_l N^{(l-1)}} \rightarrow \mathcal{M}_2^{\otimes K_l N^{(l-1)}}.\end{align}

 Thus, in total, the intermediate interfaces $\Enc_l$ and $\Dec_l$  transform between ($K_l$ copies of) level-$(l-1)$ registers and level-$l$ registers, i.e. encoding and decoding between $Q^{(l-1)}$ and $Q^{(l)}$. We note that this process corresponds to the way that the code $\mathcal{Q}^{(l)}$ is constructed from $\mathcal{Q}^{(l-1)}$, where the encoding goes from $(l-1)$ to $l$.

Then, the full interfaces for connecting $K^{(L)}$ qubits to $N^{(L)}$ qubits are constructed recursively from the level-wise intermediate interfaces as
\begin{align}
    \Enc^{(L)} &=\Enc_L \circ \qty(\Enc^{(L-1)})^{\otimes  K_L} \\&=
    \Enc_L \circ \qty(\Enc_{L-1})^{\otimes K_L} \circ \cdots \circ \qty(\Enc_2)^{\otimes K_L K_{L-1} \cdots K_3} \circ \qty(\Enc_1)^{\otimes K_L K_{L-1} \cdots K_2}
    \\&
    =\Enc_L \circ \qty(\Enc_{L-1})^{\otimes  \frac{K^{(L)}}{K^{(L-1)}}} \circ \cdots \circ \qty(\Enc_2)^{\otimes  \frac{K^{(L)}}{K^{(2)}}}  \circ \qty(\Enc_1)^{\otimes  \frac{K^{(L)}}{K^{(1)}}} 
\end{align}
and 

\begin{align}
    \Dec^{(L)} &=\qty(\Dec^{(L-1)})^{\otimes K_L}\circ\Dec_L,\\&=
    \qty(\Dec_1)^{\otimes  \frac{K^{(L)}}{K^{(1)}}} \circ \qty(\Dec_2)^{\otimes  \frac{K^{(L)}}{K^{(2)}}} \circ \cdots \circ \qty(\Dec_{L-1})^{\otimes  \frac{K^{(L)}}{K^{(L-1)}}}  \circ \Dec_L.
\end{align}

In total, $\Dec^{(L)}$ takes an $N^{(L)}$-qubit state (input to the intermediate interface $\Dec_L$) to a $K^{(L)}$-qubit state (output of $K^{(L)}$ copies of the intermediate interface $\Dec_1$, where each one outputs a single qubit), and $\Enc^{(L)}$ does the opposite. In other words, $\Dec^{(L)}$ takes a $K^{(L)}$-qubit logical state, which is encoded in $N^{(L)}$ physical qubits, to an unencoded $K^{(L)}$-qubit state, thus serving as an interface\footnote{The definitions we introduce here for interfaced circuits could, without loss of generality, also be adapted to define ``interfaced circuits with interfaces to intermediate levels'', e.g. applying a circuit encoded in $Q^{(L_1)}$ to an input state encoded in $Q^{(L_2)}$ for $L_2\leq L_1$ by constructing interfaces between (copies of) $Q^{(L_2)}$ and $Q^{(L_1)}$. In this work, we will usually refer to level-$L$ interfaced circuits as circuits with interfaces between level $0$ and level $L$.} between level $0$ and level $L$ of the code $\mc Q^{(L)}$.

We illustrate an example of an interfaced circuit for level $2$ ($L=2$) in Fig.~\ref{fig:level_l_circuit}.

\begin{figure}[h!]
    \centering
    \includegraphics[height=5cm]{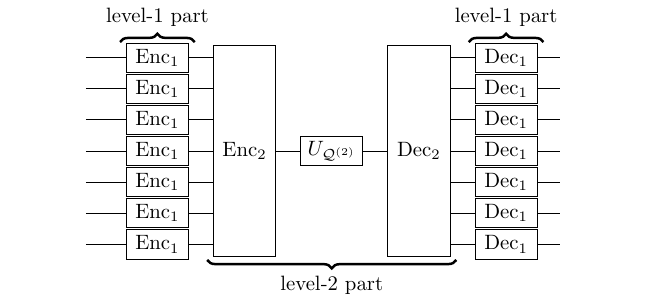}
    \caption{A level-$2$ interfaced circuit for a unitary $U$, implemented in $\mathcal{Q}^{(2)}$ and concatenated with circuits $\Enc_l$ and $\Dec_l$ mapping between the levels for $l=1,2$. The original circuit $U$ is a 7-qubit unitary, and $U_{\mc Q^{(2)}}$ is its implementation in $\mc Q^{(2)}$ in terms of a total of 105 physical qubits. We emphasize that the wires in this diagram do not correspond to physical qubits, but to registers: $\Enc_1$ takes a level-$0$ register (i.e. a physical qubit) and maps it to a level-$1$ register; $\Enc_2$ maps $K_2=7$ level-$1$ registers to a level-$2$ register, $U_{\mc Q^{(2)}}$ takes a level-$2$ register to a a level-$2$ register.}
    \label{fig:level_l_circuit}
\end{figure}

As shown in Fig.~\ref{fig:level_l_circuit}, a level-$l$ interfaced circuit may include level-$l^\prime$ parts for $l^\prime\in\{0,\ldots,l\}$.
Level-$l$ elementary operations act on level-$l$ registers in a level-$l$ part.
A level-$l$ encoding interface is a map from level-$(l-1)$ registers of a level-$(l-1)$ part into a level-$l$ register of a level-$l$ part, and a level-$l$ decoding interface from a level-$l$ register of a level-$l$ part into level-$(l-1)$ registers of a level-$(l-1)$ part.
For $l^\prime\in\{1,\ldots,l\}$, a level-$l^\prime$ part of a level-$l$ circuit is a part starting from level-$l^\prime$ preparation operations and level-$l^\prime$ encoding interfaces, ending by level-$l^\prime$ measurement operations and level-$l^\prime$ decoding interfaces, and composed of the level-$l^\prime$ elementary operations and level-$(l^\prime+1)$ parts in the middle (which may, recursively, contain level-$(l^\prime+2),\cdots,L$ parts).
A level-$0$ part refers to a circuit composed of level-$0$ (i.e., physical) operations, which may have open-end wires to input and output states of physical qubits (apart from closed-end wires starting from preparations or ending with measurements). Note that a level-$l^\prime$ part for $l^\prime\geq 1$ does not have such open-end wires. In a level-$l^\prime$ part of a level-$l$ circuit, the level-$l^\prime$ elementary operations and interfaces are called level-$l^\prime$ \textit{operation locations and interface locations}, respectively.

\subsection{Compilation for interfaced circuits}

\label{sec-compile}

For each level $l=0, \ldots, L$, we will construct a set of level-$l$ elementary operations (preparations, gates, and measurements) and, additionally, introduce level-$l$ encoding and decoding interfaces; then, we use these elementary operations and interfaces to write a level-$l$ (interfaced) circuit.

More precisely, for each level-$l$ elementary operation, we will define a corresponding level-$l$ \textit{gadget}, which is a level-$(l-1)$ circuit intended to carry out the logical elementary operation on the encoded level-$l$ registers. For each of the level-$l$ encoding and decoding interfaces, we will also define the corresponding level-$l$ encoding and decoding gadget, which is a level-$(l-1)$ circuit intended to carry out the corresponding encoding and decoding maps, respectively, between the level-$(l-1)$ registers and the encoded level-$l$ registers. We will further define a level-$l$ error-correction (EC) gadget, which is a level-$(l-1)$ circuit intended to carry out quantum error correction for an encoded level-$l$ register.

For a level-$l$ gadget $G_l$, we let $|G_l|$ denote the number of level-$(l-1)$ operation locations in the gadget.
To implement some of the gates required for universal quantum computation, our protocol may use gate teleportation that combines multiple level-$l$ elementary operations to implement a gate.
For simplicity of presentation, we may represent these multiple level-$l$ elementary operations collectively as a level-$l$ \textit{abbreviation}, which is a level-$l$ circuit intended to carry out the gate on the level-$l$ registers.
A level-$l$ circuit that includes level-$l$ abbreviations is identified with one composed only of level-$l$ elementary operations obtained by expanding all the level-$l$ abbreviations.

For the concatenated quantum Hamming codes,~\cite{yamasaki2022timeefficient} has constructed a set of level-$l$ elementary operations, gadgets, and abbreviations, which we summarize below. 
Note that the gadget for the $S^{\otimes K^{(l)}}$ gate was not present in~\cite{yamasaki2022timeefficient}, but more recently shown in~\cite{tansuwannont2025cliffordgateslogicaltransversality}. We will then introduce additional notations for level-$l$ interfaces.
The set of level-$l$ elementary operations  is denoted as follows:
\begin{align}
  \label{eq:measurement}
  &\text{measurement:}\nonumber\\&\quad\includegraphics{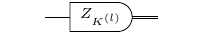},\\
  \label{eq:h_gate}
  &\text{$H$ gate:}\nonumber\\&\quad\includegraphics{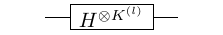},\\
  \label{eq:s_gate}
  &\text{$S$ gate:}\nonumber\\&\quad\includegraphics{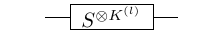},\\
  \label{eq:cnot_gate}
  &\text{\textsc{CNOT} gate:}\nonumber\\&\quad\includegraphics{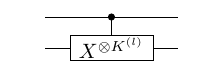},\\
  \label{eq:cz_gate}
  &\text{$CZ$ gate:}\nonumber\\&\quad\includegraphics{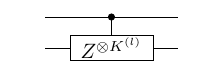},\\
  \label{eq:pauli_gate}
  &\text{Pauli gate:}\nonumber\\&\quad\includegraphics{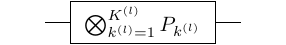},\\
  \label{eq:initial_state_preparation}
  &\text{initial-state preparation:}\nonumber\\&\quad\includegraphics{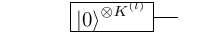},\\
  \label{eq:clifford_state_preparation}
  &\text{Clifford-state preparation:}\nonumber\\&\quad\includegraphics{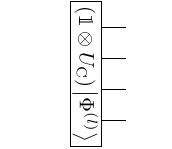},\\
  \label{eq:magic_state_preparation}
  &\text{magic-state preparation:}\nonumber\\&\quad\includegraphics{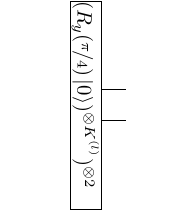}.
\end{align}
In these notations, a solid wire represents a level-$l$ register, and double-line wires represent $K^{(l)}$ classical bits.
The level-$l$ measurement operation performs measurements in $Z$ basis $\{\ket{0},\ket{1}\}$ of all the $K^{(l)}$ qubits in a level-$l$ register, outputting the $K^{(l)}$-bit measurement outcome, where we let $Z_{K^{(l)}}$ denote a label for representing this measurement.
The level-$l$ $H$-, \textsc{CNOT}-, and $CZ$-gate operations perform $H^{\otimes K^{(l)}}$, $\textsc{CNOT}^{\otimes K^{(l)}}$, and $CZ^{\otimes K^{(l)}}$ gates acting on all the qubits in level-$l$ registers.

The level-$l$ Pauli-gate operation performs a tensor product of arbitrary Pauli gates
\begin{equation}
  \label{eq:pauli_l}
  \bigotimes_{k^{(l)}=1}^{K^{(l)}}P_{k^{(l)}},
\end{equation}
where $P_{k^{(l)}}\in\{X, Z, Y,\mathds{1}\}$ is a single-qubit Pauli (or identity) gate acting on the ${k^{(l)}}$th qubit of the level-$l$ register, and the global phase is ignored.
The level-$l$ initial-state preparation operation prepares a state $\ket{0}^{\otimes K^{(l)}}$ of $K^{(l)}$ qubits in a level-$l$ register.
The level-$l$ Clifford-state preparation operation prepares a state $\left(\mathbbm{1}^{B_1 B_2}\otimes U_\mathrm{C}^{B_3 B_4}\right)\ket{\Phi^{(l)}}^{B_1 B_2 B_3 B_4}$ of four level-$l$ registers $B_1,B_2,B_3,B_4$ aligned from top to bottom in Eq.~\eqref{eq:clifford_state_preparation}, where $U_\mathrm{C}^{B_3 B_4}$ is an arbitrary Clifford unitary operator acting on qubits in $B_3$ and $B_4$,
\begin{align}
  \label{eq:phi}
  &\ket{\Phi^{(l)}}^{B_1 B_2 B_3 B_4}=\ket{\Phi}^{B_1 B_3}\otimes\ket{\Phi}^{B_2 B_4},
\end{align}
and
\begin{align}
  &\ket{\Phi}^{B_j B_{j^\prime}}=\frac{1}{\sqrt{2^{K^{(l)}}}}\sum_{m=0}^{2^{K^{(l)}}-1}\ket{m}^{B_j}\otimes\ket{m}^{B_{j^\prime}},\nonumber\\
  &\quad(j,j^\prime)\in\{(1,3),(2,4)\}.
\end{align}
Here, the superscripts represent the registers to which states and operators belong.
The level-$l$ magic-state preparation operation prepares a state ${(R_y(\nicefrac{\pi}{4})\ket{0})}^{\otimes K^{(l)}}$ of $K^{(l)}$ qubits in each of two level-$l$ registers $B_1,B_2$ aligned from top to bottom in Eq.~\eqref{eq:magic_state_preparation}.
For simplicity of notation, the level-$l$ preparation, gate, and measurement operations may be collectively labelled as $\ket{\psi_l}$, $U_l$, and $M_l$, respectively.
The notations on elementary operations may also be used to show the corresponding gadgets; in particular, whenever a level-$l$ elementary operation is depicted as the one acting on $N_{l}$ level-$(l-1)$ registers instead of each level-$l$ register, this notation represents the corresponding level-$l$ gadget.
For these gadgets, we will use the construction from~\cite{yamasaki2022timeefficient} in this work.

In addition, we introduce the level-$l$ encoding and decoding interfaces between level $l$ and level $(l-1)$
which were not previously considered in~\cite{yamasaki2022timeefficient}. 

These interfaces are represented by
\begin{align}
  \label{eq:enc}
  &\text{encoding:}\nonumber\\&\quad\includegraphics{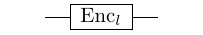},\\
  \label{eq:dec}
  &\text{decoding:}\nonumber\\&\quad\includegraphics{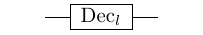}.
\end{align}
In these notations, a solid wire represents a level-$l$ register, and a thin wire collectively represents $K_l$ level-$(l-1)$ registers (which may be written as $K_l$ different wires if necessary).
The thin wire may also be written as a solid wire if the level-$(l-1)$ part of the circuit connected to the level-$l$ interface includes other level-$(l-1)$ operations. The notations on interfaces may also be used to show the corresponding gadgets; in particular, whenever a level-$l$ interface is depicted as the one between $K_{l}$ level-$(l-1)$ registers and $N_{l}$ level-$(l-1)$ registers instead of $K_{l}$ level-$(l-1)$ registers and a level-$l$ register, this notation represents the corresponding level-$l$ gadget.
The construction of these gadgets for interfaces will be given in Section~\ref{sec-gadgets}.

Apart from these elementary operations, interfaces, and the corresponding gadgets, the level-$l$ error-correction (EC) gadget is labelled as $\EC_l$ and denoted by
\begin{align}
  \label{eq:error_correction}
  &\text{error correction:}\nonumber\\
  &\quad\includegraphics{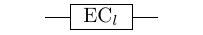}.
\end{align}
In addition, the following level-$l$ abbreviations are used:
\begin{align}
  \label{eq:clifford_abbreviation}
  &\text{two-register Clifford gate:}\nonumber\\&\quad\includegraphics{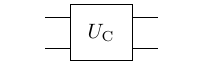},\\
  \label{eq:r_y_abbreviation}
  &\text{$R_y(\pm\nicefrac{\pi}{4})$ gate:}\nonumber\\&\quad\includegraphics{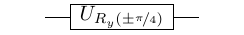},\\
  \label{eq:tr_abbreviation}
  &\text{trace:}\nonumber\\&\quad\includegraphics{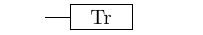},\\
  &\text{wait (arbitrary length):}\nonumber\\&\quad\includegraphics{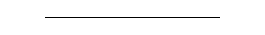},
\end{align}
which are abbreviated notations of combinations of level-$l$ elementary operations with functions described as follows.
The level-$l$ two-register Clifford-gate abbreviation applies an arbitrary Clifford unitary $U_\mathrm{C}$ on the two level-$l$ registers, which is implemented by gate teleportation in~\cite{yamasaki2022timeefficient} using level-$l$ elementary operations.
The level-$l$ $R_y(\pm\nicefrac{\pi}{4})$-gate abbreviation applies $U_{R_y(\pm\nicefrac{\pi}{4})}$ gates on the level-$l$ register, where $U_{R_y(\pm\nicefrac{\pi}{4})}$ is an arbitrary tensor product of $R_y(\nicefrac{\pi}{4})$, $R_y(-\nicefrac{\pi}{4})$, and $\mathds{1}$ and this abbreviation is implemented by gate teleportation in~\cite{yamasaki2022timeefficient} using level-$l$ elementary operations .
The level-$l$ trace abbreviation discards the input quantum state by performing level-$l$ measurement operations followed by forgetting the measurement outcome.
The level-$l$ wait abbreviation performs the identity operator on a level-$l$ register, which will be regarded as a special case of performing a sequence of the level-$l$ Pauli-gate operations of arbitrary length (from zero to many); we may also explicitly write a level-$l$ wait abbreviation as
\begin{align}
  &\text{wait (arbitrary length):}\nonumber\\&\quad\includegraphics{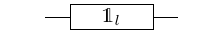}.
\end{align}
In our protocol, we use the construction in~\cite{yamasaki2022timeefficient} for the EC gadgets and the abbreviations.

\begin{figure}[htbp]
    \centering
    \includegraphics[scale=0.79]{figures/Compile-new-1.pdf}\\
    \hspace{0.5cm}\\
    {\Large{$\Downarrow$}}\\
    \hspace{0.5cm}\\
    \includegraphics[scale=0.79]{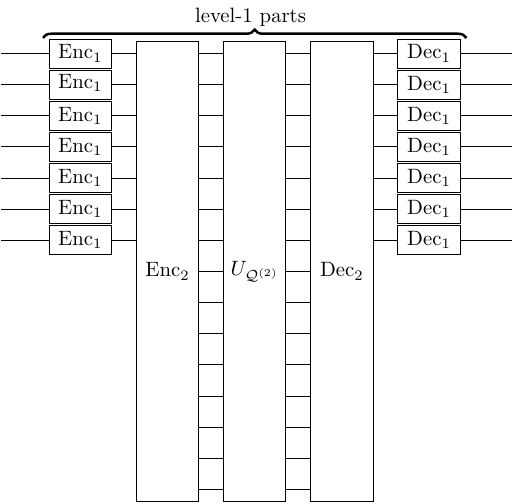}\\
    \hspace{0.5cm}\\
    {\Large{$\Downarrow$}}\\
    \hspace{0.5cm}\\
    \includegraphics[scale=0.79]{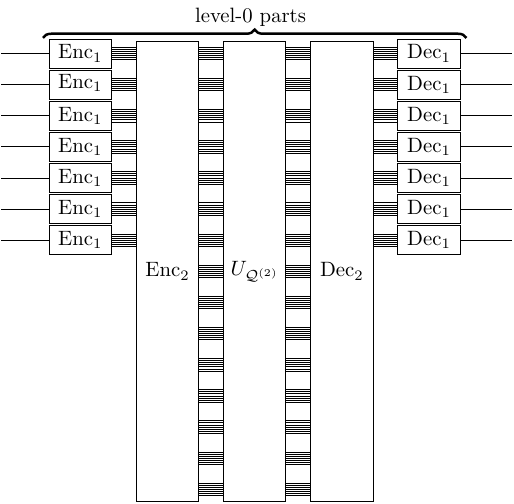}
    \caption{Compilation of a level-$2$ interfaced circuit $U$ ($L=2$), where we use the $[[7,1,3]]$ code at concentenation level $1$, and the $[[15,7,3]]$ code at level $2$. For each $l\in\{L,\ldots,1\}$, we recursively compile the level-$l$ circuit into the level-$(l-1)$ circuit by replacing every level-$l$ operation and interface location in the level-$l$ parts of the level-$l$ circuit with its corresponding level-$l$ gadget and by inserting the level-$l$ EC gadgets between all the adjacent pairs of two level-$l$ operations in $U^{Q^{(L)}}$. Then, the level-$1$ version of the circuit is obtained as the level-$2$ circuit circuit written in terms of level-$1$ gadgets and level-$1$ interfaces. Taking the interfaced circuit from Fig.~\ref{fig:level_l_circuit}, $U_{\mc Q^{(2)}}$ in the first figure is understood as a unitary on level-$2$ registers (i.e. the wires going in and out of $U_{\mc Q^{(2)}}$ represent one level-$2$ register); in the second picture, it represents a unitary on $15$ level-$1$ registers; in the third picture, it represents the unitary on the total 105 physical qubits. These unitaries and interfaces at lower levels are obtained from the higher levels by the compilation procedure outlined below.
    }
    \label{fig:circuit_conversion}
\end{figure}

Using these gadgets, we now present a recursive procedure to compile the level-$l$ interfaced circuit into a level-$(l-1)$ interfaced circuit for each $l=L,L-1,\ldots,1$, as illustrated in Fig.~\ref{fig:circuit_conversion} for L=2.
Let $U:\mc M_2^{\otimes K^{(l)}} \to \mc M_2^{\otimes K^{(l)}}$ be some quantum circuit. Then, $\Dec^{(l)} \circ U_{Q^{(l)}} \circ \Enc^{(l)}$ denotes the level-$l$ interfaced circuit.
Note that 
\[\Dec^{(l)} \circ U_{Q^{(l)}} \circ \Enc^{(l)}=(\Dec^{(l-1)})^{\otimes K_l} \circ  (\Dec_l \circ U_{Q^{(l)}} \circ \Enc_l) \circ (\Enc^{(l-1)} )^{\otimes K_l}.\]
The output of $(\Enc^{(l-1)})^{\otimes K_l}$ is given by $K_l N^{(l-1)}$ physical qubits, which by $\Enc_l$ are transformed into $N^{(l)}$ physical qubits, on which $U_{\mc Q^{(l)}} $ acts.

Now we use the following compilation procedure to compile the level-$l$ circuit $(\Dec_l \circ U_{Q^{(l)}}  \circ \Enc_l)$ in terms of a level-$(l-1)$ circuit.
For each level-$l$ register in the level-$l$ circuit, we use, in the corresponding level-$(l-1)$ circuit, a set of $K_{l}$ level-$(l-1)$ registers and further add a constant number of auxiliary level-$(l-1)$ registers per set for implementing the gadgets while maintaining the constant space overhead as in~\cite{yamasaki2022timeefficient}.
The corresponding level-$(l-1)$ circuit is given by replacing every level-$l$ operation and interface location in the level-$l$ parts of the level-$l$ circuit with its corresponding level-$l$ gadget, followed by inserting the level-$l$ EC gadgets between all the adjacent pairs of gadgets of two level-$l$ operations (see Fig.~\ref{fig:circuit_conversion}).
Note that the level-$l$ EC gadgets are not inserted between adjacent pairs of gadgets of a level-$l$ operation and a level-$l$ interface, nor between adjacent pairs of two level-$l$ interfaces (e.g., $\Enc_l$ followed immediately by $\Dec_l$)\footnote{To show that EC gadgets are not necessary around interfaces, we will introduce and use the correctness conditions $\textbf{Enc+U+EC}$ and $\textbf{EC+U+Dec}$ in Eqs.~\eqref{eq:correct-u-and-enc} and~\eqref{eq:correct-u-and-dec}.}.
To maintain the constant space overhead, we also insert waits before and after the EC gadgets, as described in~\cite{yamasaki2022timeefficient}. As a result, we obtain a level-$(l-1)$ interfaced circuit: this circuit acts on $K_l$ level-$(l-1)$ registers, and has $K_l$ interfaces $\Enc^{(l-1)}$ and $\Dec^{(l-1)}$ between level $0$ and level $(l-1)$.

By performing this procedure recursively for each level, we obtain the level-$0$ interfaced circuit.
Just after performing the above recursive procedure, the level-$0$ circuit is composed of the level-$0$ abbreviations and the level-$0$ elementary operations, which may use auxiliary level-$0$ registers by definition.
However, at level $0$, we are allowed to perform operations directly on physical qubits.
Thus, we substitute the level-$0$ two-register Clifford-gate abbreviations in the level-$0$ circuit with direct applications of the two-qubit Clifford gates and the level-$0$ $R_y(\pm\nicefrac{\pi}{4})$-gate abbreviations with the one-qubit $R_y(\pm\nicefrac{\pi}{4})$ gates, using the operations on physical qubits rather than performing the gate teleportation.
After this substitution for the level-$0$ abbreviations, we also substitute the remaining level-$0$ elementary operations with the corresponding operations on the physical qubits.
As a result, we do not use the auxiliary level-$0$ registers.
Hence after these substitutions, we remove the auxiliary level-$0$ registers from the level-$0$ circuit, which yields the fault-tolerant circuit on physical qubits to be executed in our fault-tolerant protocol.

\subsection{Level-conversion theorem for interfaced circuits}

In this section, we define properties of the parts of a circuit implemented in the concatenated quantum Hamming code that ensure that the implementation fulfills a notion of ``fault-tolerance''. For non-interfaced circuits, these concepts (properties of the gadgets or the associated extended rectangles) are well-known in the literature but sometimes appear under subtly different names. For this reason, and for the convenience of the reader, we summarize briefly in Table~\ref{tab:theproperties} how we call the various properties in this work and where they are defined explicitly in the following sections.
\begin{table}[h!]
    \centering
    \begin{tabular}{l|l|l}
    \hline
        Condition & Can apply to & Meaning \\ \hline
      fault-tolerant & gate gadget & fulfills the transformation rules from Eqs.~\eqref{eq-ft-cond1}-\eqref{eq-ft-cond2} \\
        valid & interface gadget & fulfills the transformation rules from Eqs.~\eqref{eq:validenc}-\eqref{eq:validdec} \\
        good & gate ExRec & has less than $t_l$ faults ($\sum_i s_i\leq t_l$)\\
        good &interface gadget &has no faults ($s=0$) \\correct& ExRec (gate or interface)& fulfills the transformation rules from Eqs.~\eqref{eq:correctness_prep_dec}-\eqref{eq:correctness_enc}\\
    \hline
    \end{tabular}
    \caption{Summary of the names and where to find the conditions we use to show a level-conversion theorem for interfaced circuits in the concatenated quantum Hamming code. The details are given in Sections~\ref{sec-gadgets} and \ref{sec-correctness-conditions}. In our analysis, we will later make use of these conditions by specifying when and how they imply one another. (For example, we will make use of the fact that a good ExRec is correct, that a good truncated ExRec is correct, and that a good and valid interface gadget is correct. These relations between the conditions play a crucial role in the proof of Theorem~\ref{thm-level-conversion} (the level-conversion theorem), which is one of the main results of this work.)
}
    \label{tab:theproperties}
\end{table}

\subsubsection{Fault tolerance conditions for interfaced circuits}
\label{sec-gadgets}

We present requirements for fault tolerance and a construction of the level-$l$ gadgets for level-$l$ elementary operations and interfaces for each concatenation level $l\in\{0,1,\ldots,L\}$. 

As for the level-$l$ interfaces,~\cite{CMH20} also provides a construction of gadgets for interfaces, but we here present a different, simpler construction for the corresponding gadgets; moreover, we also discuss a refined set of their required properties for the analysis of their fault tolerance.

We define the required properties of gadgets for fault tolerance.
To define such requirements, as in the existing analyses of fault-tolerant protocols with concatenated codes~\cite{G,yamasaki2022timeefficient}, we introduce a $*$-decoder and $r$-filters for our protocol.
In~\cite{G,yamasaki2022timeefficient}, the conditions for fault-tolerant gadgets are given in the form of equivalence relations between circuits including the gadget, the $r$-filters, the $*$-decoder, and the ideal intended elementary operation, conditioned on the number of faulty locations in the gadget.
In particular,~\cite{G} provides concrete definitions for the $*$-decoder and the $r$-filters for $[[N, 1, 2t + 1]]$ codes for $0\leqq r\leqq t$, and~\cite{yamasaki2022timeefficient} provides those for a $[[N, K, 3]]$ codes; by contrast, we here provide a general definition for $[[N, K, 2t + 1]]$ codes.
The proof of the threshold theorem in~\cite{G,yamasaki2022timeefficient} does not rely on specific definitions but on the equivalence relations alone; thus, we will give appropriate definitions of a $*$-decoder and $r$-filters so that the gadgets proposed here should satisfy effectively the same set of equivalence relations as those in~\cite{G,yamasaki2022timeefficient}.
In addition, we further introduce the notion of a $*$-encoder, an inverse of the $*$-decoder, which will be useful for our analysis.

We begin by defining the $r$-filter for the $[[N_l,K_l,D_l=2t_l+1]]$ code $\mathcal{Q}_{l}$ for $r\leq t_l$.
A level-$l$ register is encoded into $N_{l}$ level-$(l-1)$ registers,
and the $k^{(l-1)}$th qubit of the $n$th level-$(l-1)$ register is labeled $(n,k^{(l-1)})$ with $k^{(l-1)}\in\{1,\ldots,K^{(l-1)}\}$ and $n\in\{1,\ldots,N_{l}\}$.
Let
\begin{equation}
  \mathcal{H}_{n,k^{(l-1)}}\cong\mathbb{C}^2
\end{equation}
be the Hilbert space of the qubit $(n,k^{(l-1)})$ in a level-$(l-1)$ register.
Define
\begin{align}
  \mathcal{H}_{k^{(l-1)}}&\coloneqq\bigotimes_{n=1}^{N_{l}} \mathcal{H}_{n,k^{(l-1)}},\\
  \label{eq:H}
  \mathcal{H}&\coloneqq\bigotimes_{k^{(l-1)}=1}^{K^{(l-1)}} \mathcal{H}_{k^{(l-1)}},
\end{align}
where $\mathcal{H}$ is the whole space of the $N_{l}$ level-$(l-1)$ registers.
The quantum code $\mathcal{Q}_{l}$ defines a code space for each $k^{(l-1)}$
\begin{equation}
  \mathcal{H}_{k^{(l-1)}}^{\mathrm{code}} \subset \mathcal{H}_{k^{(l-1)}}.
\end{equation}
Then the whole code space in the $N_{l}$ level-$(l-1)$ registers is given by
\begin{equation}
  \mathcal{H}^{\mathrm{code}}\coloneqq\bigotimes_{k^{(l-1)}=1}^{K^{(l-1)}} \mathcal{H}_{k^{(l-1)}}^{\mathrm{code}} \subset \mathcal{H}.
\end{equation}
The level-$l$ $0$-filter acting on the $N_{l}$ level-$(l-1)$ registers is defined as a projector
\begin{equation}
  \label{eq:0filter}
  \Pi_l^{(0)}~\text{onto the subspace $\mathcal{H}^{\mathrm{code}}\subset\mathcal{H}$}.
\end{equation}
Suppose that a codeword $\ket{\psi_{k^{(l-1)}}} \in \mathcal{H}_{k^{(l-1)}}^{\mathrm{code}}$ suffers from an error represented by a Pauli operator $P_{k^{(l-1)}} \in \mathcal{P}_{N_{l}}$ acting on $\mathcal{H}_{k^{(l-1)}}$.
If the weight of $P_{k^{(l-1)}}$ is at most $r$, the code $\mathcal{Q}_{l}$ can correct the error.
As a whole, if a codeword $\ket{\psi} \in \mathcal{H}^{\mathrm{code}}$ suffers from an error represented by $P=\bigotimes_{k^{(l-1)}=1}^{K^{(l-1)}} P_{k^{(l-1)}}$ with the weight of each $P_{k^{(l-1)}}$ at most $r$, the $K^{(l-1)}$ blocks of the code $\mathcal{Q}_{l}$ can correct the error $P$.
For $r>0$, the level-$l$ $r$-filter $\Pi^{(r)}$ acting on the $N_{l}$ level-$(l-1)$ registers is defined as a projector 
\begin{equation}
  \label{eq:rfilter}
    \Pi_l^{(r)}~\text{onto the subspace spanned by all the states in the form $P\ket{\psi}$,}
\end{equation}
where each $P_{k^{(l-1)}}$ for $P=\bigotimes_{k^{(l-1)}=1}^{K^{(l-1)}} P_{k^{(l-1)}}$ has a weight at most $r$, and $\ket{\psi} \in \mathcal{H}^{\mathrm{code}}$.
We write a level-$l$ $r$-filter as
\begin{equation}
    \includegraphics{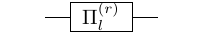}.
\end{equation}

For $K^{(l-1)}$ blocks of the $[[N_l,K_l,D_l=2t_l+1]]$ code $\mathcal{Q}_{l}$ used at concatenation level $l$, the level-$l$ $*$-decoder \[\Dec_{l}^*: \big( \mathcal{M}_2^{\otimes N^{(l-1)}}\big)^{\otimes N_l} \to \big( \mathcal{M}_2^{\otimes N^{(l-1)}} \big)^{\otimes K_l} \otimes\big( \mathcal{M}_2^{\otimes N^{(l-1)}}\big)^{\otimes (N_l-K_l)}\] is defined as a Clifford unitary transformation mapping a level-$l$ register  into $K_l$ level-$(l-1)$-registers and a syndrome state on the remaining $(N_l-K_l)$ $(l-1)$-registers. Note that, as a map, this means that $\Dec^*_l$ acts on a total number of $N^{(l)}=N_lN^{(l-1)}$ qubits (which are $N_l$ $l-1$-registers, i.e. $N_l$ copies of $K^{(l-1)}$ logical qubits encoded in $N^{(l-1)}$ physical qubits) and outputs a total of $K_lN^{(l-1)}$ qubits (which are $K_l$ copies of  $K^{(l-1)}$ logical qubits encoded in $N^{(l-1)}$ physical qubits) and the syndrome state, which is an $(N_l-K_l)N^{(l-1)}$-qubit state.

We depict the level-$l$ $*$-decoder as
\begin{equation}
\label{eq:stardec}
    \includegraphics{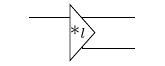},
\end{equation}
where the left solid wire represents the $N_l$ level-$(l-1)$ registers, the top-right thin wire represents the level-$l$ register, and the bottom-right thin wire represents the syndrome state living on $(N_l-K_l)$ $(l-1)$-registers.
The thin wires may also be written as solid wires if their meaning is obvious from the context.

Similarly, we define the level-$l$ $*$-encoder 
\[\Enc_{l}^*:\big( \mathcal{M}_2^{\otimes N^{(l-1)}} \big)^{\otimes K_l} \otimes\big( \mathcal{M}_2^{\otimes N^{(l-1)}}\big)^{\otimes (N_l-K_l)}\to  \big( \mathcal{M}_2^{\otimes N^{(l-1)}}\big)^{\otimes N_l} \]
as the inverse Clifford unitary transformation of the level-$l$ $*$-decoder, which is depicted as
\begin{equation}
\label{eq:starenc}
    \includegraphics{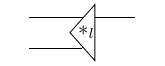}.
\end{equation}
A level-$l$ ideal decoder refers to the level-$l$ $*$-decoder with tracing out the syndrome qubits, and a level-$l$ ideal encoder refers to the level-$l$ $*$-encoder with inputting zero syndromes $\ket{\boldsymbol{0}}=\ket{0}\otimes\ket{0}\otimes\cdots\otimes\ket{0}$ to all the syndrome qubits.
The level-$l$ ideal decoder and encoder are depicted in the same way as Eq.~\eqref{eq:stardec} and Eq.~\eqref{eq:starenc}, respectively, with the thin wires representing the syndrome qubits omitted.
By definition, the $*$-decoder and the $*$-encoder satisfy the identity relations
\begin{align}
\label{eq:idenA}
    \includegraphics{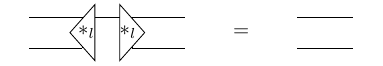},\\
\label{eq:idenB}
    \includegraphics{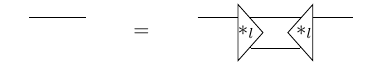}.
\end{align}

In this work, we will use the following notation for circuits subject to faults according to the local stochastic Pauli error model.
For any level-$l$ gadget, we will use a superscript to represent the number of faulty level-$(l-1)$ locations in the gadget.
For example, for a level-$l$ preparation operation labeled $\ket{\psi_l}$, the state of level-$(l-1)$ registers prepared by the corresponding gadget with $s$ faulty locations is denoted by $\ket{\psi_l}^{(s)}$.
For a level-$l$ gate operation $U_l$, the state transformation on level-$(l-1)$ registers implemented by the corresponding gadget with $s$ faulty locations is denoted by $U_l^{(s)}$.
In particular, a sequence of wait operations including $s$ faulty locations is denoted by $\mathds{1}_l^{(s)}$.
For a level-$l$ measurement operation $M_l$, the measurement of level-$(l-1)$ registers implemented by the corresponding gadget with $s$ faulty locations is denoted by $M_l^{(s)}$.
For level-$l$ encoding and decoding interfaces $\Enc_l$ and $\Dec_l$, the state transformations on level-$(l-1)$ registers implemented by the corresponding gadgets with $s$ faulty locations are denoted by $\Enc_l^{(s)}$ and $\Dec_l^{(s)}$, respectively.

Using these notations, for level-$l$ operations with the $[[N_l,K_l,D_l=2t_l+1]]$ code $\mathcal{Q}_{l}$, we define the fault tolerance conditions on the gadgets and subsequently define that for the level-$l$ EC gadget, which are given in such a way that the argument in~\cite{G,yamasaki2022timeefficient} carries over.
Roughly speaking, when a gadget includes many faulty locations in general, the states of the logical and syndrome qubits obtained from the output state of the gadget via the $*$-decoder may be affected.
We represent such an effect of errors by a CPTP map $\mathcal{S}$.
Under the local stochastic Pauli error model, if the overall circuit is a stabilizer circuit, the CPTP map $\mathcal{S}$ would be a random Pauli channel; however, in the presence of non-Clifford gates, $\mathcal S$ may be a more general CPTP map and need not be Pauli. Because errors at faulty locations in the local stochastic Pauli error model do not involve quantum memory in the environment, this general CPTP map nonetheless still acts only on syndrome or logical qubits.

In a slight abuse of notation, in the definitions below, $\mathcal{S}$ does not always refer to the same map in each definition; in fact, the map $\mathcal{S}$ may sometimes refer to a map on syndrome qubits for one code block (for the gadgets for preparation, measurement, and single-register gates) or to a map on syndrome qubits for two code blocks (for the gadgets for two-register gates).
 
The fault-tolerant conditions of the gadgets require that if the number of faulty locations is below $t_l$, the weight of the errors does not spread, and the logical qubits obtained from the $*$-decoder should remain correctable and thus unaffected by $\mathcal{S}$ while the syndrome qubits may still be affected.

In particular, a level-$l$ preparation gadget preparing the state $\ket{\psi_l}$ is fault-tolerant if it satisfies the following two properties:
\begin{align}\label{eq-ft-cond1}
    &\text{\textbf{Prep A\@:} when $s\leq t_l$}\nonumber\\
    &\includegraphics{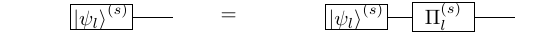},\\
    &\text{\textbf{Prep B\@:} when $s\leq t_l$}\nonumber\\
    &\includegraphics{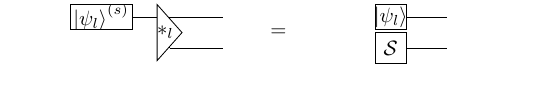},
\end{align}
where $\mathcal{S}$ is some quantum channel acting on the syndrome qubits that does not affect the logical state $\ket{\psi_l}$.
For a level-$l$ preparation gadget for multiple encoded registers, the fault-tolerant conditions are defined in the same way, e.g., in the case of two encoded registers
\begin{align}
    &\text{\textbf{Prep A\@:} when $s\leq t_l$}\nonumber\\
    &\includegraphics{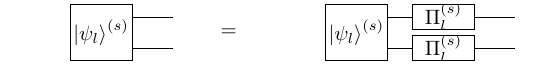},\\
    &\text{\textbf{Prep B\@:} when $s\leq t_l$}\nonumber\\
    &\includegraphics{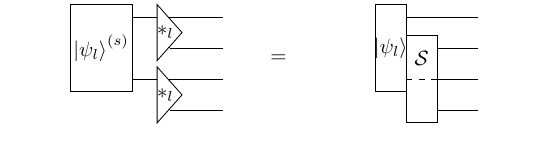},
\end{align}
where $\mathcal{S}$ is some quantum channel acting on two syndrome registers that does not affect the logical state.
A level-$l$ single-register gate gadget $U_l$ is fault-tolerant if it satisfies the following two properties: 

\begin{align}
    &\text{\textbf{Gate A\@:} when $s+r\leq t_l$}\nonumber\\
    &\includegraphics{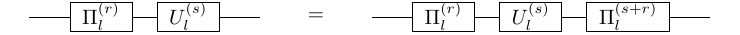},\\
    &\text{\textbf{Gate B\@:} when $s+r\leq t_l$}\nonumber\\
    &\includegraphics{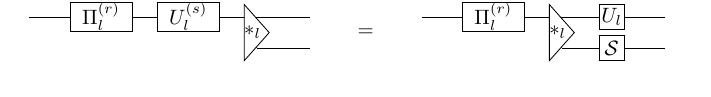}.
\end{align}
For a level-$l$ two-register gate gadget, the fault-tolerant conditions are defined as
\begin{align}
    &\text{\textbf{Gate A\@:} when $s+\sum_i r_i\leq t_l$}\nonumber\\
    &\includegraphics{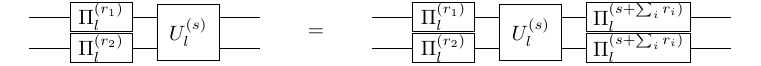},\\
    &\text{\textbf{Gate B\@:} when $s+\sum_i r_i\leq t_l$}\nonumber\\
    &\includegraphics{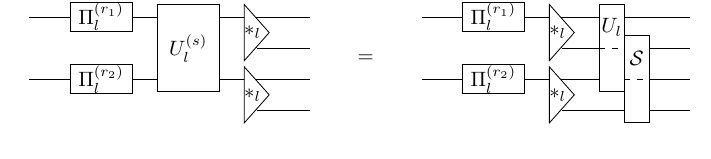},
\end{align}
where $\sum_i r_i$ is the sum over all the input encoded registers.
A level-$l$ measurement gadget $M_l$ is fault-tolerant if it satisfies the following property:
\begin{align}
    &\text{\textbf{Meas\@:} when $r+s\leq t_l$}\nonumber\\
    &\includegraphics{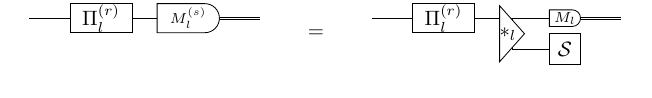}.
\end{align}
A level-$l$ EC gadget $\mathrm{EC}_l$ is fault-tolerant if it satisfies the following two properties:
\begin{align}
    &\text{\textbf{EC A\@:} when $s\leq t_l$}\nonumber\\
    &\includegraphics{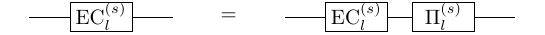},\\
    &\text{\textbf{EC B\@:} when $r+s\leq t_l$}\nonumber\\
    &\includegraphics{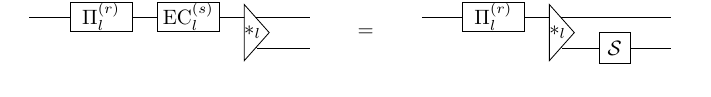}.\label{eq-ft-cond2}
\end{align}

For level-$l$ interfaces, we newly introduce the required properties of the corresponding gadgets.
The properties introduced here are different from the conditions in the previous work~\cite{CMH20,BCMH24} in that the level-$l$ gadgets for the level-$l$ interfaces here are composed of level-$(l-1)$ operations, and thus, their required properties are given in terms of those imposed on the level-$(l-1)$ operations while~\cite{CMH20,BCMH24} directly analyzes the level-$0$ (i.e., physical) implementations.
Our definition makes it possible to follow the conventional level-conversion argument in~\cite{G,yamasaki2022timeefficient}.
We call the required conditions on the level-$l$ gadgets for the interfaces \textit{validity conditions}. 
A level-$l$ encoding gadget $\mathrm{Enc}_l$ is valid if it satisfies the following property:
\begin{align}
\label{eq:validenc}
    &\text{\textbf{Enc\@:}}\nonumber\\
    &\includegraphics{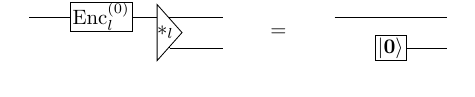}
\end{align}
where $\mathrm{Enc}_l^{(0)}$ has no faulty level-$(l-1)$ location, and $\ket{\boldsymbol{0}}=\ket{0}\otimes\ket{0}\otimes\cdots\ket{0}$ sets all syndrome qubits to zero. 
Also, a level-$l$ decoding gadget $\mathrm{Dec}_l$ is valid if it satisfies the following property:
\begin{align}
\label{eq:validdec}
    &\text{\textbf{Dec\@:}}\nonumber\\
    &\includegraphics{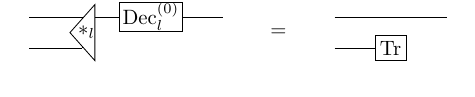}
\end{align}
where $\mathrm{Dec}_l^{(0)}$ has no faulty level-$(l-1)$ location.

\begin{figure}[t]
    \centering
    \includegraphics[width=7.0in]{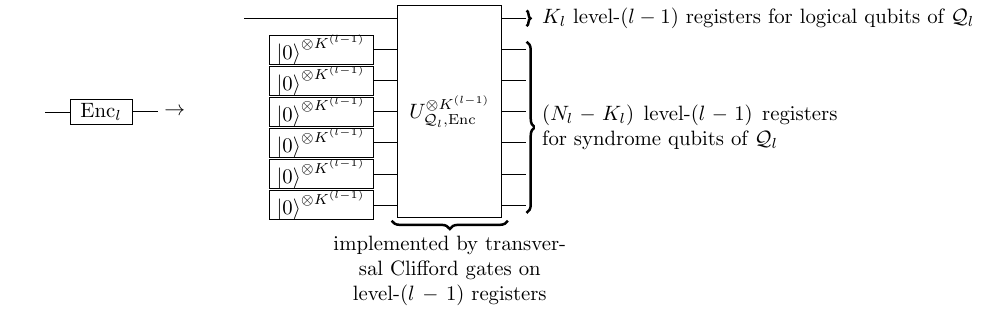}
    \includegraphics[width=7.0in]{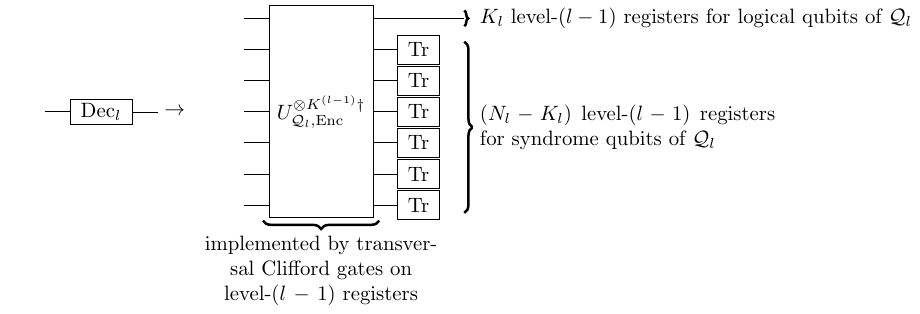}
    \caption{The level-$l$ encoding gadget at the top and the level-$l$ decoding gadget at the bottom.}
    \label{fig:interfaces}
\end{figure}

The explicit construction of the fault-tolerant gadgets in terms of the universal gate set for preparation, measurement, single-register gates, 2-register gates and EC is given in~\cite{yamasaki2022timeefficient}, and our protocol here uses the same gadgets for those operations.
As for the interfaces, we give the construction of level-$l$ encoding and decoding gadgets in Fig.~\ref{fig:interfaces}.
In particular, the construction in Fig.~\ref{fig:interfaces} uses the unitary
\begin{equation}
\label{eq:enc_unitary}
    U_{\mathcal{Q}_l,\mathrm{Enc}}^{\otimes K^{(l-1)}},
\end{equation}
where for any $K_l$-qubit state $\ket{\psi}$, $U_{\mathcal{Q}_l,\mathrm{Enc}}$ transforms $U_{\mathcal{Q}_l,\mathrm{Enc}}(\ket{\psi}\otimes\ket{0}^{\otimes(N_l-K_l)})=\ket{\overline\psi}$, and $\ket{\overline\psi}$ is an $N_l$-qubit state representing the logical state $\ket{\psi}$ of the code $\mathcal{Q}_l$. 
As shown in~\cite{PhysRevA.56.76,yamasaki2022timeefficient}, this $N_l$-qubit Clifford unitary $U_{\mathcal{Q}_l,\mathrm{Enc}}$ can be implemented by a stabilizer circuit composed of at most $O(N_l(N_l-K_l))$ $H$, \textsc{CNOT}, $CZ$, and $Z$ gates, i.e., with $O(N_l(N_l-K_l))$ depth on $N_l$ qubits.
Since we have $N_l-K_l=O(\log(N_l))$ for the quantum Hamming codes. the depth of the circuit implementing $U_{\mathcal{Q}_l,\mathrm{Enc}}$ is $O(N_l\log(N_l))$.
Then, the $K^{(l-1)}$ copies of the unitary in Eq.~\eqref{eq:enc_unitary} used in the construction of Fig.~\ref{fig:interfaces} are also implementable by a $O(N_l\log(N_l))$-depth circuit composed of level-$(l-1)$ $H$-, \textsc{CNOT}-, $CZ$-, and Pauli-gate operations acting on $K^{(l-1)}$ qubits of each of the $N_l$ level-$(l-1)$ registers.
Since these unitaries are dominant in  the construction of Fig.~\ref{fig:interfaces}, the number of level-$(l-1)$ operations in each of the level-$l$ encoding and decoding gadgets is bounded by
\begin{align}
\label{eq:dec_size}
    |\Dec_l|=O\qty(N_l^2\log(N_l)),\\
\label{eq:enc_size}
    |\Enc_l|=O\qty(N_l^2\log(N_l)).
\end{align}

\subsubsection{Correctness conditions for interfaced circuits}
\label{sec-correctness-conditions}

Using the gadgets defined in Section~\ref{sec-gadgets}, we obtain the fault-tolerant circuit from the level-$L$ circuit using the compilation procedure described in Fig.~\ref{fig:circuit_conversion}, and we will later prove level-conversion and threshold theorems for this interfaced circuit to study the effective error rates.
For this purpose, as in the conventional analysis of the threshold theorem for protocols with concatenated codes in~\cite{G,yamasaki2022timeefficient}, we count the number of faults in the extended rectangles, or ExRecs for short. 
As introduced in~\cite{yamasaki2022timeefficient}, given a level-$l$ circuit and the corresponding level-$(l-1)$ circuit, for each level-$l$ operation location of the level-$l$ circuit, the level-$l$ ExRec corresponding to the level-$l$ operation location is defined as a part of the level-$(l-1)$ circuit consisting of the level-$l$ gadget corresponding to the level-$l$ location, all the level-$l$ EC gadgets placed between the location and the adjacent locations, and the level-$(l-1)$ wait abbreviations inserted between these gadgets.
A difference from~\cite{G} arises from the fact that our protocol, as well as the protocol in~\cite{yamasaki2022timeefficient}, inserts level-$(l-1)$ wait abbreviations between the gadgets to synchronize the timing of logical gates.
Due to this difference, a level-$l$ ExRec for a level-$l$ location in our protocol includes the level-$(l-1)$ wait operations between the level-$l$ gadget for the location and each adjacent level-$l$ EC gadget (see also~\cite{yamasaki2022timeefficient} for illustration).
For each level-$l$ ExRec at a level-$l$ operation location, the level-$l$ EC gadgets appearing before the level-$l$ gadget of the level-$l$ location are called leading EC gadgets, and those appearing after are called trailing EC gadgets.

Following~\cite{G,yamasaki2022timeefficient}, we define the correctness of the level-$l$ ExRecs\@.
Roughly speaking, the correct level-$l$ ExRecs implement the logical operations correctly, where the noise at faulty locations affects only the syndrome qubits, and the logical qubits are manipulated as intended.
In our analysis, we introduce two notions of correctness; one is in terms of the level-$l$ $*$-decoder, and the other is in terms of the level-$l$ $*$-encoder.
Whereas the ones using the $*$-decoder are immediately obtained from the conventional definition in~\cite{G} and are also used for analyzing the previous protocol in~\cite{yamasaki2022timeefficient}, the additional ones using the $*$-encoder are also useful for our analysis, as we will see later.

A level-$l$ preparation ExRec is correct for some particular arrangements of faults if there exists a quantum channel $\mc S$ such that
\begin{align}
\label{eq:correctness_prep_dec}
    &\text{\textbf{Prep Dec$^*$\@:}}\nonumber\\
    &\includegraphics{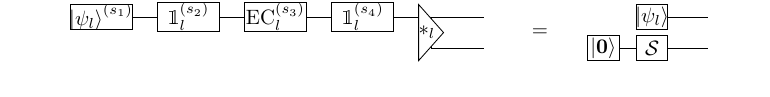},\\
\label{eq:correctness_prep_enc}
    &\text{\textbf{Prep Enc$^*$\@:}}\nonumber\\
    &\includegraphics{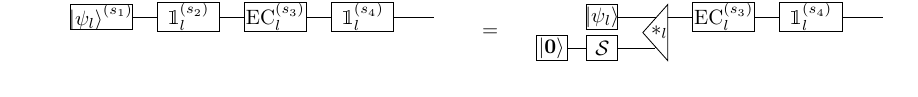},
\end{align}
where $\mathds{1}_l$ represents the wait locations in the ExRec inserted before and after the EC gadgets in our protocol.
For a level-$l$ preparation ExRec for multiple encoded registers, the correctness is defined in the same way, e.g., in the case of two encoded registers:
\begin{align}
\label{eq:correctness_prep2_dec}
    &\text{\textbf{Prep Dec$^*$\@:}}\nonumber\\
    &\includegraphics{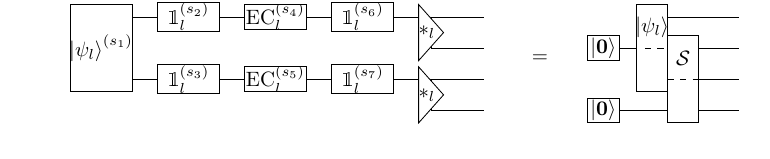},\\
\label{eq:correctness_prep2_enc}
    &\text{\textbf{Prep Enc$^*$\@:}}\nonumber\\
    &\includegraphics{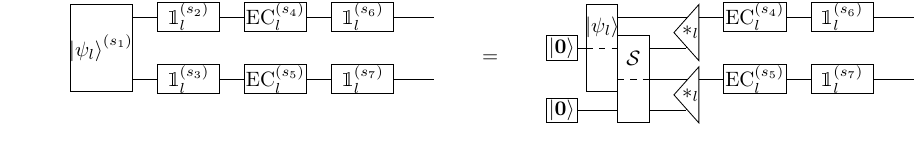}.
\end{align}
A level-$l$ gate ExRec is correct for some particular arrangements of faults if
\begin{align}
\label{eq:correctness_gate_dec}
    &\text{\textbf{Gate Dec$^*$\@:}}\nonumber\\
    &\includegraphics[width=7.0in]{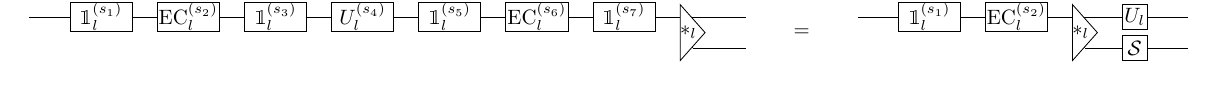},\\
\label{eq:correctness_gate_enc}
    &\text{\textbf{Gate Enc$^*$\@:}}\nonumber\\
    &\includegraphics[width=7.0in]{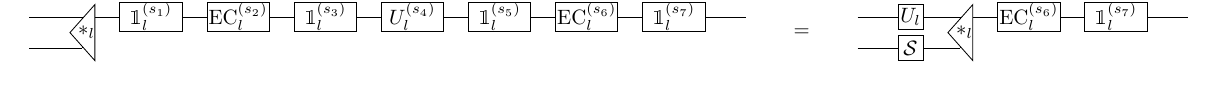}.
\end{align}
For a level-$l$ gate ExRec for multiple encoded registers, the correctness is defined in the same way, e.g., in the case of two encoded registers: 

\begin{align}
\label{eq:correctness_gate2_dec}
    &\text{\textbf{Gate Dec$^*$\@:}}\nonumber\\
    &\includegraphics[width=7.0in]{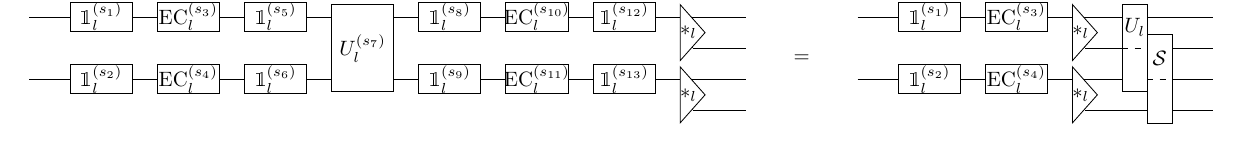},\\
\label{eq:correctness_gate2_enc}
    &\text{\textbf{Gate Enc$^*$\@:}}\nonumber\\
    &\includegraphics[width=7.0in]{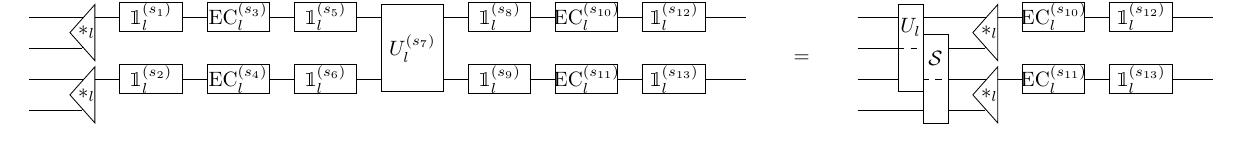}.
\end{align}
A level-$l$ measurement ExRec is correct for some particular arrangements of faults if
\begin{align}
\label{eq:correctness_meas_dec}
    &\text{\textbf{Meas Dec$^*$\@:}}\nonumber\\
    &\includegraphics{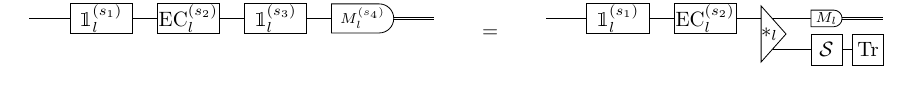},\\
\label{eq:correctness_meas_enc}
    &\text{\textbf{Meas Enc$^*$\@:}}\nonumber\\
    &\includegraphics{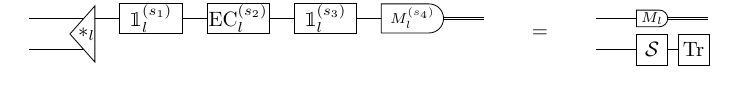}.
\end{align}

In the same way as the correctness of the level-$l$ ExRecs, we also introduce the notion of correctness of the level-$l$ interface gadgets.
A level-$l$ decoding gadget is correct for some particular arrangements of faults if
\begin{align}
\label{eq:correctness_dec}
    &\text{\textbf{Dec\@:}}\nonumber\\
    &\includegraphics{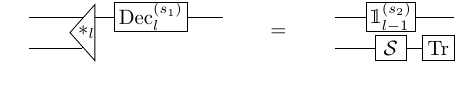}.
\end{align}

A level-$l$ encoding gadget is correct for some particular arrangements of faults if
\begin{align}
\label{eq:correctness_enc}
    &\text{\textbf{Enc\@:}}\nonumber\\
    &\includegraphics{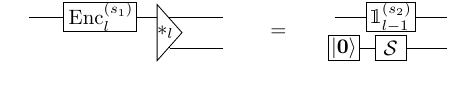}.
\end{align}

The level-$l$ combination of a circuit $U$ and an interface is correct for some arrangement of faults if

\begin{align}
\label{eq:correct-u-and-enc}
    &\text{\textbf{Enc+U+EC\@:}}\nonumber\\
    &\includegraphics{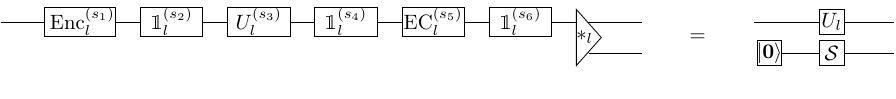},
\end{align}

\begin{align}
\label{eq:correct-u-and-dec}
    &\text{\textbf{EC+U+Dec\@:}}\nonumber\\
    &\includegraphics{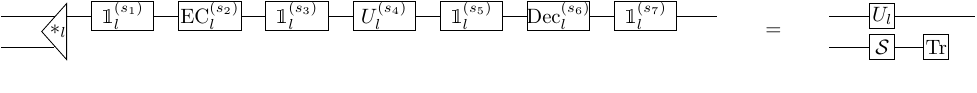}.
\end{align}

With these definitions of the correctness of the level-$l$ ExRecs and interface gadgets,
the error analysis of the fault-tolerant circuit reduces to the analysis of the conditions on the correctness of each ExRecs and interface gadgets for the $[[N_l,K_l,D_l=2t_l+1]]$ code.
A level-$l$ ExRec is said to be good if it contains at most $t_l$ faults, i.e., in Eqs.~\eqref{eq:correctness_prep_dec}--\eqref{eq:correctness_meas_enc}
\begin{equation}
    \sum_i s_i\leq t_l,
\end{equation}
and bad otherwise~\cite{G}.
Regarding the correctness defined for the level-$l$ $*$-decoder, i.e., Eqs.~\eqref{eq:correctness_prep_dec},~\eqref{eq:correctness_prep2_dec},~\eqref{eq:correctness_gate_dec},~\eqref{eq:correctness_gate2_dec}, and ~\eqref{eq:correctness_meas_dec}, it is proven that a good ExRec is correct, as shown in~\cite{G} by moving the $*$-decoder in these definitions using the fault-tolerance conditions.
As for the correctness defined for the level-$l$ $*$-encoder, we can also prove that a good ExRecs is correct in the sense of Eqs.~\eqref{eq:correctness_prep_enc},~\eqref{eq:correctness_prep2_enc},~\eqref{eq:correctness_gate_enc},~\eqref{eq:correctness_gate2_enc}, and~\eqref{eq:correctness_meas_enc}.
In particular, Eqs.~\eqref{eq:correctness_prep_enc},~\eqref{eq:correctness_prep2_enc},~\eqref{eq:correctness_gate_enc}, and~\eqref{eq:correctness_gate2_enc} is obtained by inserting the identity on the right-hand side of Eq.~\eqref{eq:idenB} before each of the trailing EC gadgets of Eqs.~\eqref{eq:correctness_prep_dec},~\eqref{eq:correctness_prep2_dec},~\eqref{eq:correctness_gate_dec},~\eqref{eq:correctness_gate2_dec}, moving the $*$-decoder by the fault-tolerance conditions in the same way, and cancelling out the $*$-decoder with the $*$-encoder using Eq.~\eqref{eq:idenA}.

If we have a bad ExRec, the error represented by $\mathcal{S}$ may occur on the logical qubits and the syndrome qubits; e.g., in the case of the gate ExRec, we have
\begin{align}
    &\text{\textbf{Gate Dec$^*$\@:}}\nonumber\\
    &\includegraphics[width=7.0in]{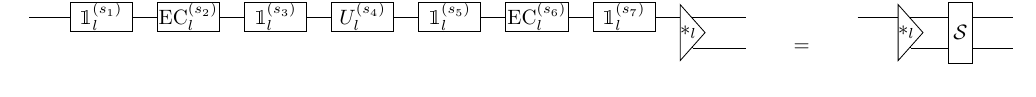},\\
    &\text{\textbf{Gate Enc$^*$\@:}}\nonumber\\
    &\includegraphics[width=7.0in]{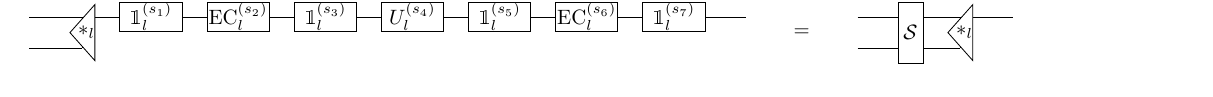}.
\end{align}

As discussed in~\cite{G}, if we have a bad ExRec while moving the $*$-decoder, such bad ExRecs may eliminate the trailing EC gadgets of the previous ExRecs\@.
In the same way, if we have a bad ExRec while moving the $*$-encoder, bad ExRecs may eliminate the leading EC gadgets of the later ExRecs\@.
We call these ExRecs missing one or more EC gadgets \emph{truncated ExRecs}.
A truncated ExRec is said to be good if it contains at most $t_l$ faults, and bad otherwise.
As shown in~\cite{G}, a good truncated ExRec is correct, and the probability of a truncated ExRec being bad is upper bounded by the corresponding (full) ExRec without truncation, which holds regardless of moving the $*$-decoder as in~\cite{G} or the $*$-encoder as in our case.

The error analysis of the interface gadgets uses the validity condition in place of the fault-tolerance conditions.
In particular, we say that a level-$l$ interface gadget is good if it contains no fault, i.e., in Eq.~\eqref{eq:correctness_dec} and Eq.~\eqref{eq:correctness_enc}
\begin{equation}
    s=0,
\end{equation}
and bad otherwise.
Then by definition, a good level-$l$ decoding gadget satisfying the validity condition Eq.~\eqref{eq:validdec} is correct in the sense of Eq.~\eqref{eq:correctness_dec}, and a good level-$l$ encoding gadget satisfying the validity condition Eq.~\eqref{eq:validenc} is correct in the sense of Eq.~\eqref{eq:correctness_enc}. 
If interface gadgets are bad, we have
\begin{align} \label{eq:interfacesbad}
    &\text{\textbf{Dec\@:}}\nonumber\\
    &\includegraphics{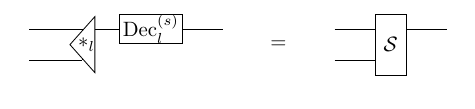},\\
    &\text{\textbf{Enc\@:}}\nonumber\\
    &\includegraphics{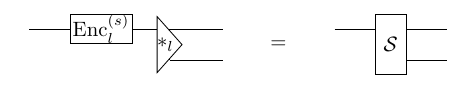}.
\end{align}

Similarly, if the combination of an interface and a gate is bad, we have
\begin{align} \label{eq:u+interfacesbad}
    &\text{\textbf{Dec\@:}}\nonumber\\
    &\includegraphics{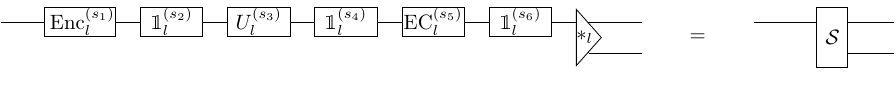},\\
    &\text{\textbf{Enc\@:}}\nonumber\\
    &\includegraphics{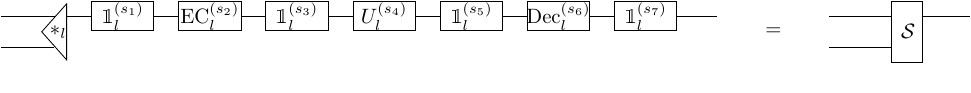}.
\end{align}

\subsubsection{Level-conversion theorem and its consequences}
\label{sec-level-conversion}

We are now in a position to prove a level-conversion theorem bounding the effective error rates of an interfaced circuit implemented in the concatenated quantum Hamming code.

\begin{figure}[h!]
\centering
    \includegraphics[scale=0.7]{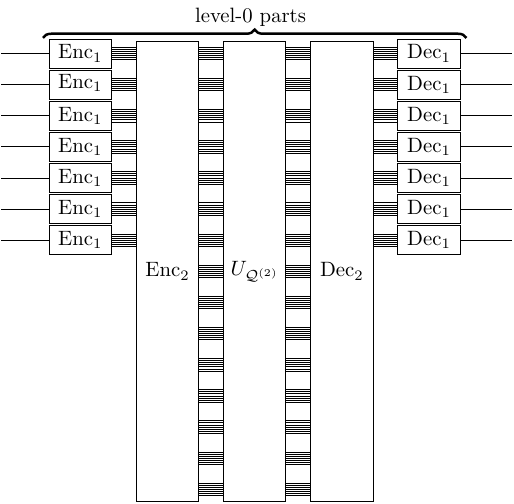} 
    \raisebox{3cm}{\hbox{\Large{$\Rightarrow$}}}
    \includegraphics[scale=0.7]{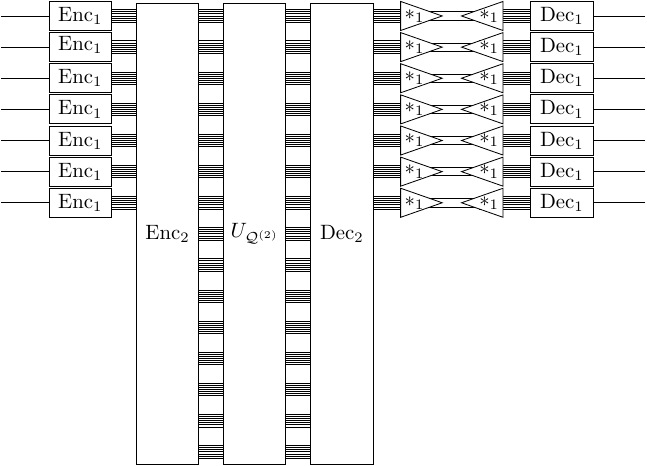}\\
    \raisebox{2cm}{\hbox{\Large{$\Rightarrow$}}}
    \raisebox{1cm}{\includegraphics[scale=0.7]{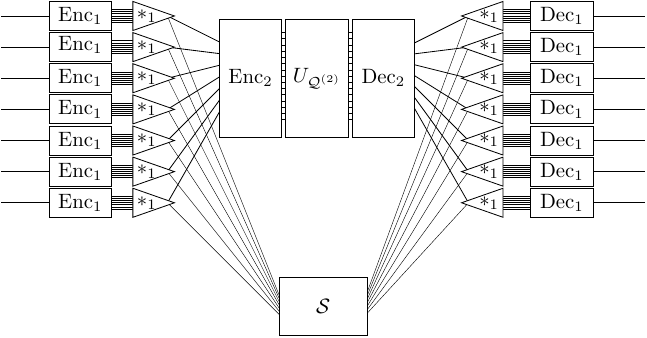}}
    \raisebox{2cm}{\hbox{\Large{$\Rightarrow$}}}
    \includegraphics[scale=0.7]{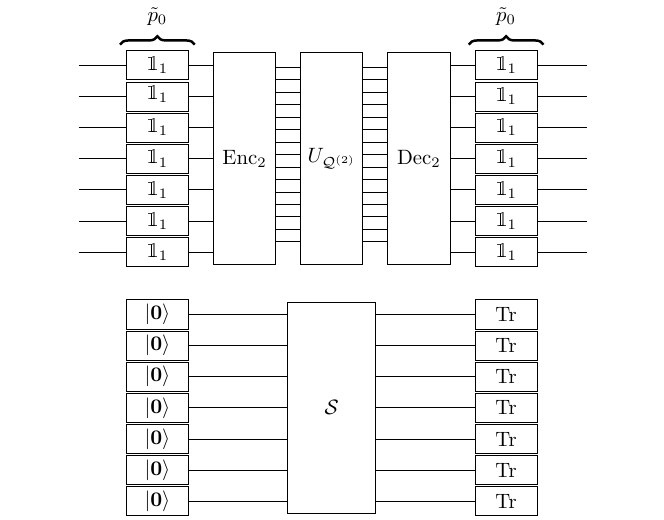}\\
    \raisebox{2cm}{\hbox{\Large{$\Rightarrow$}}}
    \includegraphics[scale=0.9]{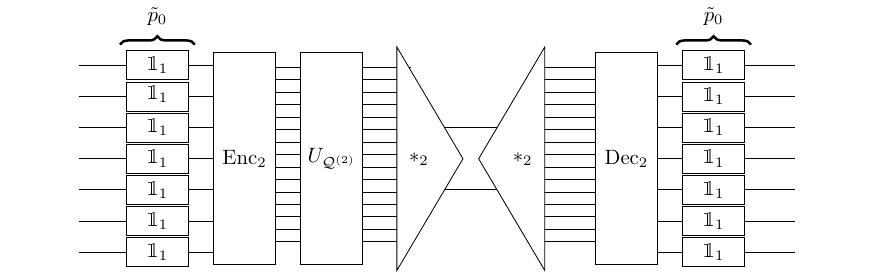}\\
    \raisebox{2cm}{\hbox{\Large{$\Rightarrow$}}}
    \includegraphics[scale=0.9]{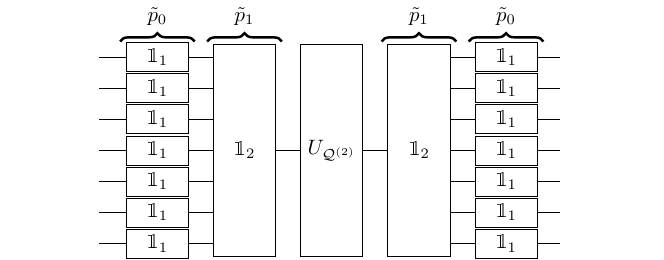}
    \caption{Level conversion to identify the error rate of each of the level-$l$ operations and interfaces ($l\in\{0,1,\ldots,L\}$) in the intended level-$L$ circuit implemented by the fault-tolerant circuit. For each level $l\in\{1,\ldots,L\}$, the level conversion of the interfaced circuit for our protocol is performed by inserting a pair of the $*$-decoder and the $*$-encoder before every level-$l$ decoding gadget $\Dec_l$ and every level-$l$ measurement gadget. The inserted $*$-encoders and $*$-decoders are moved through the level-$(l-1)$ circuit using the conditions on the correctness of the level-$l$ ExRecs until they are eliminated by the correctness conditions of the level-$l$ ExRecs or cancelled out. Then, the obtained level-$l$ circuit undergoes the local stochastic Pauli error model, and the error rate of each level-$l$ operation and interface location is upper bounded according to Theorem~\ref{thm-level-conversion}.
    }
    \label{fig:level_conversion}
\end{figure}

\begin{theorem}[Level conversion] \label{thm-level-conversion}
Let $\mathcal{Q}^{(L)}$ be the $[[N^{(L)},K^{(L)},D^{(L)}]]$ concatenated quantum Hamming code with concatenation level $L$ and threshold $p_{\textrm{th}}$.

Let $U:\mathcal{M}_2^{\otimes K^{(L)}}\to\mathcal{M}_2^{\otimes K^{(L)}}$ be a $K^{(L)}$-qubit quantum circuit, and let $U_{\mathcal{Q}^{(L)}}:\mathcal{M}_2^{\otimes N^{(L)}}\to\mathcal{M}_2^{\otimes N^{(L)}}$ denote its implementation in $\mathcal{Q}^{(L)}$ (in terms of physical qubits). Let $\Dec^{(L)}\circ U_{\mathcal{Q}^{(L)}} \circ \Enc^{(L)} :\mathcal{M}_2^{\otimes K^{(L)}}\to\mathcal{M}_2^{\otimes K^{(L)}}$ denote the corresponding interfaced circuit.

Suppose that the interfaced circuit is affected by local stochastic Pauli noise with noise parameter $p_0$. Then, for any $p_0\leq p_{\textrm{th}}$, there exists a level-$l$ interfaced implementation of the same circuit that is subject to local stochastic Pauli noise, with a universal error rate $\tilde{p}_l$ with
\begin{equation}
\label{eq:error_rate_universal}
    \tilde{p}_l\leq\qty(\frac{p_0}{\tilde{p}_\mathrm{th}})^{2^l}\tilde{p}_\mathrm{th},
\end{equation}
such that the error rate of each level-$l$ operation is upper bounded by $\tilde{p}_l$ and the error rate of each level-$l$ interface
 is upper bounded by $\tilde{p}_{l-1}$.
\end{theorem}

\begin{proof}
\textbf{Level conversion.}
We perform the level conversion as in the conventional analysis in~\cite{G}  using the $*$-decoder, as shown in Fig.~\ref{fig:level_conversion}.

Given a fault-tolerant circuit obtained from a level-$L$ circuit, we recursively perform a level-conversion procedure for $l\in\{1,2,\ldots,L\}$.
In the $l$-th step, we find parts of the circuit corresponding to the level-$l$ parts of the original level-$L$ circuit by using the compilation procedure introduced in Section~\ref{sec-compile}.
Then, we insert the right-hand side of the identity Eq.~\eqref{eq:idenB}
before every $\Dec_l$ interface and every level-$l$ measurement operation.
The inserted $*$-decoders are moved through the level-$l$ circuit using the conditions on the correctness of the level-$l$ ExRecs until they are eliminated by the correctness conditions of the level-$l$ preparation ExRecs (i.e., Eqs.~\eqref{eq:correctness_prep_dec},~\eqref{eq:correctness_prep2_dec}, and~\eqref{eq:correctness_meas_enc}) and the level-$l$ interface gadgets (Eqs.~\eqref{eq:correctness_dec} and~\eqref{eq:correctness_enc}), or cancelled out by the identity Eq.~\eqref{eq:idenA}.
If the level-$l$ ExRecs are good, they are correct and are replaced with the logical operations as intended; also, if the level-$l$ interface gadgets are good, they are correct and are replaced with the wait operation as shown in Eqs.~\eqref{eq:correctness_dec} and~\eqref{eq:correctness_enc}.
Due to the conventional argument in~\cite{G}, with this procedure, the level-$l$ parts undergo the local stochastic Pauli error model if the level-$(l-1)$ parts do. This holds by induction under our assumption that the level-$0$ (physical) part undergoes the local stochastic Pauli error model. Moreover, the error rate of each level-$l$ operation and interface location is upper bounded by the probability of having a bad level-$l$ ExRec and interface gadget, respectively, corresponding to the location\footnote{This argument holds not only for the local stochastic Pauli error model in our setting but also generally for the local stochastic error model as in~\cite{G,yamasaki2022timeefficient}, where the noise can be arbitrary CPTP maps rather than Pauli errors.}.

\textbf{Bounding the error rates.}
To bound the error rate of each level-$l$ operation, as in~\cite{yamasaki2022timeefficient}, we recall that our protocol uses a concatenation of the $[[N_l=2^{l+2}-1,K_l=2^{l+2}-2(l+2)-1,3]]$ quantum Hamming code at each concatenation level $l$.
Then, as shown in~\cite{yamasaki2022timeefficient}, the logical error rate of each level-$l$ operation in the level-$l$ parts is upper bounded by
\begin{align}
    \label{eq:error_rate} p_l = \text{Prob}(\text{A level-$l$ operation is not correct})\leq \qty(\frac{p_0}{p_\mathrm{th}})^{2^l}p_\mathrm{th}
\end{align}
for some threshold constant $p_\mathrm{th}>0$.

Here, in addition to the error rate of the level-$l$ operations, we also need to clarify the error rate of the level-$l$ interfaces.

We first discuss the decoding interface.
Since the good level-$l$ interface gadgets composed of the level-$(l-1)$ operations at error rate $p_{l-1}$ in Eq.~\eqref{eq:error_rate} are correct, the error rates of the level-$l$ decoding interface are upper bounded by as follows\footnote{Notably, the interface from the physical system to the first level of concatenation $\Enc_1$ fails with a probability proportional to the physical error rate $p_0$. Unfortunately, this error cannot be suppressed arbitrarily and is an effect of considering circuits with quantum input and output, as previously observed in~\cite{gottesman2014faulttolerant}, \cite[Lemma~11]{8555154},~\cite{MGHHLPP14} and \cite[Theorem~III.3]{CMH20}. This is precisely the reason why an interfaced circuit with quantum input or output has different behaviour in the context of fault tolerance.}:
\begin{align}
    \text{Prob}(\text{$\Dec_l^{(s)}$ is not correct})&\leq \text{Prob}(\text{$\Dec_l^{(s)}$ is not good})\\
    &\leq |\Dec_l| \times \text{Prob}(\text{A level $(l-1)$-operation is not correct})\\
    &\leq |\Dec_l|p_{l-1}.
\end{align}

Due to Eq.~\eqref{eq:dec_size}, for some constant $\alpha>0$, we have
\begin{equation}
\label{eq:dec_size_bound}
    |\Dec_l|\leq \mathrm{poly}(N_l)\leq 2^{\alpha l},
\end{equation}
where we use $N_l=2^{l+2}-1$.
Therefore, from Eq.~\eqref{eq:error_rate} and Eq.~\eqref{eq:dec_size_bound}, we obtain the following upper bound:
\begin{equation}
\label{eq:error_rate_decoding}
    |\Dec_l|p_{l-1}\leq\qty(\frac{p_0}{p_\mathrm{th}})^{2^{l-1}}p_\mathrm{th}2^{\alpha l}.
\end{equation}
and thus have
\begin{align}
    \text{Prob}(\text{$\Dec_l^{(s)}$ is not correct})&\leq   \qty(\frac{p_0}{p_\mathrm{th}})^{2^{l-1}}p_\mathrm{th}2^{\alpha l}.
\end{align}

In the same way, using Eq.~\eqref{eq:enc_size} in place of Eq.~\eqref{eq:dec_size}, we also obtain an upper bound of the error rate of the level-$l$ encoding interface
\begin{equation}
\label{eq:error_rate_encoding}
    |\Enc_l|p_{l-1}\leq \qty(\frac{p_0}{p_\mathrm{th}})^{2^{l-1}}p_\mathrm{th}2^{\alpha^\prime l},
\end{equation} 
where $\alpha^\prime>0$ is some constant.

For simplicity, we summarize Eqs.~\eqref{eq:error_rate},~\eqref{eq:error_rate_decoding}, and~\eqref{eq:error_rate_encoding} (taking the largest upper bound and rescaling depending on $\alpha$ or $\alpha'$): In total, we have a universal threshold $\tilde{p}_\mathrm{th}>0$ and a universal error rate $\tilde{p}_l$ scaling as
\begin{equation}
    \tilde{p}_l\leq\qty(\frac{p_0}{\tilde{p}_\mathrm{th}})^{2^l}\tilde{p}_\mathrm{th},
\end{equation}
such that the error rate of each level-$l$ operation is upper bounded by $\tilde{p}_l$ and the error rate of each level-$l$ interface
 is upper bounded by $\tilde{p}_{l-1}$.
\end{proof}

Using this level-conversion theorem, we will formulate a threshold theorem for interfaced circuit where the noisy interfaces are understood as effective error channels. To begin, we present a simpler version of a threshold theorem for a circuit that starts and ends in an intermediate interface mapping between level $l$ and level $l-1$:

\begin{corollary}[Single-level interfaced circuit]\label{thm-single-level-threshold}
Let $U:\mc M_2^{\otimes K^{(l)}} \to \mc M_2^{\otimes K^{(l)}}$ be some quantum circuit, and let $U_{l}:\mc M_2^{\otimes N^{(l)}} \to \mc M_2^{\otimes N^{(l)}}$ be the implementation of $U$ in $Q^{(l)}$. Note that $\Dec_l \circ U_{l}  \circ \Enc_l:\mc M_2^{\otimes K_l N^{(l-1)}}\to \mc M_2^{\otimes K_l N^{(l-1)}}$. Let $\tilde{U}_{l}:\mc M_2^{\otimes K_l N^{(l-1)}}\to \mc M_2^{\otimes K_l N^{(l-1)}}$ be the unitary that implements $U_{l}$ in terms of $K_l$ $(l-1)$-registers.
Then, there exists a quantum state $\sigma_S\in \mc M_2^{\otimes (N_l-K_l) N^{(l-1)}}$ and quantum channels $\mc V_1:\mc M_2^{\otimes K_l N^{(l-1)}}\to \mc M_2^{\otimes K_l N^{(l-1)}}$ and $\mc V_2:\mc M_2^{\otimes N^{(l)}}\to \mc M_2^{\otimes K_l N^{(l-1)}}$ such that
\begin{align}
   \| [ \Dec_l \circ U_{l} \circ \Enc_l ]_{\mc F}-  \mc R_{\mathrm{dec},l} \circ \big( ( \tilde{U}_{l} \circ \mc R_{\mathrm{enc},l} ) \otimes \mc \sigma_S \big)  \|_{1\to 1} \leq \tilde{p}_l |\mathrm{Loc}(U)|
\end{align}
with
\begin{align}
   & \mc R_{\mathrm{enc},l}= (1-\tilde{p}_{l-1})\id_{K_l N^{(l-1)}} +  \tilde{p}_{l-1} \mc V_1,\\& \mc R_{\mathrm{dec},l} = (1-\tilde{p}_{l-1}) \id_{K_l N^{(l-1)}}\otimes \Tr_{(N_l-K_l) N^{(l-1)}} +\tilde{p}_{l-1} \mc V_2.
\end{align}
and $\tilde{p}_{l}$ given by Eq.~\eqref{eq:error_rate_universal}.
\end{corollary}

\begin{proof}
We insert the $*$-encoders and $*$-decoders:
    \begin{align}
    [ \Dec_l \circ U_l \circ \Enc_l ]_{\mc F}  & =[ \Dec_l \circ \Enc_l^* \circ \Dec_l^* \circ U_l \circ \Enc_l ]_{\mc F} \\ &= [\Dec_l \circ \Enc_l^*]_{\mc F_1 } \circ [\Dec_l^* \circ U_l \circ \Enc_l ]_{\mc F_2}
\end{align}

Assume $U_l$ ends in $\EC_l$ and that $U_l$ is correct under the fault pattern, i.e. every operation in $U_l$ is correct. Then, using the transformation rules \eqref{eq:correctness_dec} and \eqref{eq:interfacesbad}, we have
\begin{align} \label{eq-effchannelstep0}
    [\Dec_l \circ \Enc_l^*]_{\mc F_1 } = (1-q) \id_{l-1}\otimes \mc (\Tr \circ \mc S_1) + q \tilde{\mc V}_1
\end{align}
where $\mc S_1$ is some quantum channel that is followed by a trace, and $q$ is the probability that $\Dec_{l}$ is not correct which is upper bounded by $\tilde{p}_{l-1}$ according to Theorem~\ref{thm-level-conversion}. We thus also have that there exists some $\mc V_1$ such that \begin{align} \label{eq-effchannelstep1}
    [\Dec_l \circ \Enc_l^*]_{\mc F_1 } = (1-q) \id_{l-1}\otimes \mc (\Tr \circ \mc S_1) + q \tilde{\mc V}_1 = (1-\tilde{p}_{l-1}) \id_{l-1}\otimes \mc (\Tr \circ \mc S_1) + \tilde{p}_{l-1} \mc V_1,
\end{align}
where $\mc V_1=\qty(1/\tilde{p}_{l-1})\qty(q\tilde{\mc V}_1+\qty(\tilde{p}_{l-1}-q)\id_{l-1}\otimes \mc (\Tr \circ \mc S_1))$.

Using the transformation rules \eqref{eq:correctness_enc} and \eqref{eq:interfacesbad}, we further have 
\begin{align} \label{eq-effchannelstep2}
    [\Dec_l^* \circ U_{l} \circ \Enc_l ]_{\mc F_2} = (1-r) \tilde{U}_{l} \otimes \mc S_2 +r \tilde{ \mc V}_2,
\end{align}
where $\mc S_2$ is some preparation channel and $r$ is the probability that $\Enc_l+U_{l}$ is correct. If $U_{l}$ is correct, this probability is equal to the probability that $\Enc_l$ is correct, which is upper bounded by $\tilde{p}_{l-1}$ according to Theorem~\ref{thm-level-conversion}.
Combining \eqref{eq-effchannelstep1} and \eqref{eq-effchannelstep2}, we have some quantum state $\sigma_S$ obtained by $\mc S_2$ followed by $\mc S_1$ (i.e., some noisy preparation followed by more noise).

The above transformation rule only applies if $U_l$ is indeed correct, which is true with probability upper bounded by $\tilde{p_l}\Loc(U)$ according to Theorem~\ref{thm-level-conversion} and the union bound.
\end{proof}

In other words, we have shown that, if the decoding interface is correct, $\Dec_l$ acts like an identity from the code space at level $l$ to the level $l-1$ --- if not, it may act in a highly convoluted way as some channel that we call $\tilde{\mathcal{V}}_2$ (and a corresponding statement is true for the encoding interface as well). For an interfaced circuit implementation between level $L$ and the physical level $l=0$, which is composed of various level-$l$ parts and layers of interfaces between the levels, we can apply a similar procedure. The intermediate interfaces act as effective channels that act as an identity channel between the levels with some probability $\tilde{p}_{l-1}$, and in some convoluted way otherwise. This convoluted way is related to a concatenation of the errors from Eq.~\eqref{eq:interfacesbad} at various levels and the gates contained in the implementation of $U_{L}$ at various levels, but in our analysis, we will make no further assumptions about its structure.

\begin{corollary}[Level conversion for interfaced circuits]\label{thm-threshold-interfaced}
    Let $\mathcal{Q}^{(L)}$ be the concatenated quantum Hamming code with concatenation level $L$.
    There exists a threshold $\tilde{p}_\mathrm{th}>0$ such that, for any physical error rate $p_0$ with $0\leq p_0< \tilde{p}_\mathrm{th}$, any concatenation level $L\in\mathbbm{N}$, and any quantum circuit $U:\mathcal{M}_2^{\otimes K^{(L)}}\rightarrow \mathcal{M}_2^{\otimes K^{(L)}}$, there exist quantum channels $\mathcal{V}_{\mathrm{enc},l}:\mathcal{M}_2^{\otimes N^{(l)}}\rightarrow \mathcal{M}_2^{\otimes N^{(l)}}$ and $\mathcal{V}_{\mathrm{dec}, l}:\mathcal{M}_2^{\otimes N^{(l)}} \rightarrow \mathcal{M}_2^{\otimes K^{(l)}}$ for each $l\in[1,L]$, and a quantum state $\sigma_S\in \mathcal{M}_2^{\otimes (N^{(L)}-K^{(L)})}$
    such that
    \begin{align}\label{eq-error-term-threshold-thm}
    &\left\| [\Dec^{(L)} \circ  U_{Q^{(L)}}\circ \Enc^{(L)}  ]_{\mathcal{F}(p_0)}  - \mathcal{N}_\mathrm{dec} \circ 
   \qty( (U \circ \mathcal{N}_\mathrm{enc})\otimes \sigma_S )\right\|_{1\rightarrow 1} \leq \qty(\qty(\frac{p_0}{\tilde{p}_\mathrm{th}})^{2^{L}}\tilde{p}_\mathrm{th})|\Loc(U)| .\end{align} Here, $ \mc N_{\mathrm{dec}}:M_2^{\otimes N^{(L)}}\to \mc M_2^{\otimes K^{(L)}}$ is given by
    \begin{align}
        \mc N_{\mathrm{dec}}&=  \mc N_{\mathrm{dec},1}^{\otimes {K^{(L)}}}\circ  \mc N_{\mathrm{dec},2}^{\otimes \frac{K^{(L)}}{K^{(2)}}}\circ \cdots \circ \mc N_{\mathrm{dec},L-1}^{\otimes K_L }\circ \mc N_{\mathrm{dec},L} ,
        \end{align}
   where, for each $l$, the channels $\mc N_{\mathrm{dec},l}:M_2^{\otimes N^{(l)}}\to \mc M_2^{\otimes K_l N^{(l-1)}}$ are given by
\begin{align}
       \mc N_{\mathrm{dec},l} &= (1-\tilde{p}_{l-1}) \id_{K^{(l)}}\otimes \Tr_{(N^{(l)}-K^{(l)})} + \tilde{p}_{l-1} \mc V_{\mathrm{dec},l}.
       \end{align}
    The channel $\mc N_{\mathrm{enc}}: \mc M_2^{\otimes  K^{(L)}}\to \mc M_2^{\otimes N^{(L)}}$ is given by
    \begin{align}
        \mc N_{\mathrm{enc}}&= \mc N_{\mathrm{enc},L} \circ \mc N_{\mathrm{enc},L-1}^{\otimes K_L}\circ \cdots \circ \mc N_{\mathrm{enc},2}^{\otimes \frac{K^{(L)}}{K^{(2)}}} \circ \mc N_{\mathrm{enc},1}^{\otimes \frac{K^{(L)}}{K_1}},
\end{align}
   where $\mc N_{\mathrm{enc},l}:M_2^{\otimes K_l N^{(l-1)}}\to \mc M_2^{\otimes K_l N^{(l-1)}}$ is given by
   \begin{align}
      \mc N_{\mathrm{enc},l} &= (1-\tilde{p}_{l-1}) \id_{K_l N^{(l-1)}}  + \tilde{p}_{l-1} \mc V_{\mathrm{enc},l}
    \end{align}
and $\tilde{p}_l$ is given by Eq.~\eqref{eq:error_rate_universal}.
\end{corollary}

\begin{proof}[Proof sketch.]
    This statement follows from a recursive application of Corollary \ref{thm-single-level-threshold}. Again, we insert $*$-encoders and decoders between the unitary $U_L$ and the interfaces. Suppose that $U_L$ is correct; this is not the case with probability $\leq \tilde{p}_L \Loc(U)$, giving rise to the error term in Eq.~\eqref{eq-error-term-threshold-thm}. If $U_L$ is correct, we apply the transformation rules of the interfaces to obtain effective channels. For $\Dec_L$ and $\Enc_L$, we use the transformation rules from \eqref{eq:correct-u-and-enc} (combined with $U_L$), and for the lower levels, we use the transformation rules from \eqref{eq:correctness_enc} and \eqref{eq:correctness_dec}. The encoding interface, if correct, acts as an identity on $K_l N^{(l-1)}$ qubits and some preparation channel on the remaining $(N_l-K_l) N^{(l-1)}$ qubits; grouping all the preparation channels for the levels $l=0,1,...,L-2,L-1$ together and concatenating them with the syndrome channel from the correctness condition for $U_L$ and $\Enc_L$ gives rise to the state $\sigma_S$.
\end{proof}

\section{Fault-tolerant entanglement-assisted channel coding with concatenated quantum Hamming codes}
\label{sec-coding-thms}

In this section, we present upper and lower bounds on the achievable rates for fault-tolerant entanglement-assisted communication by constructing explicit fault-tolerant coding schemes for general channels using the concatenated quantum Hamming code. As we will illustrate in Figs.~\ref{fig:capacity-comp1}, we further compare our results to previous work, demonstrating that an implementation with this code allows us to obtain substantially higher achievable rates.

We begin by restating the following simple upper bound on the achievable communication rates for fault-tolerant entanglement-assisted communication.
From this upper bound, it is already clear that faulty gates in the encoding and decoding circuit prevent us from completely recovering the noiseless case unless the local gate error $p_0$ approaches zero.

\begin{theorem}[{\cite[Theorem~6.1]{BCMH24}}]\label{thm-cap-upper-bound} Let $p_0\geq0$. For every quantum channel $\mathcal{N}:\mathcal{M}_{d_A}\rightarrow \mathcal{M}_{d_B}$, we have 
    \[C^\mathrm{ea}_{\mathcal{F}(p_0)}(\mathcal{N}) \leq(1-p_0) C^\mathrm{ea}(\mathcal{N}).\]
\end{theorem}

Here, we show the following lower bound:

\begin{theorem} \label{thm-cap-lower-bound}
For any quantum channel $\mathcal{N}:\mathcal{M}_{2}^{\otimes j_1} \rightarrow \mathcal{M}_{2}^{\otimes j_2}$ and any local gate error rate $p_0<\tilde{p}_\mathrm{th}$ below the threshold $\tilde{p}_\mathrm{th}$ in Corollary~\ref{thm-threshold-interfaced}, we have
\begin{align}
    C_{\mathcal{F}(p_0)}^\mathrm{ea}(\mathcal{N}) \geq C^\mathrm{ea}(\mathcal{N})-16p_0(j_1+j_2)j_2 - 2 (1+4(j_1+j_2)p_{0})h_2\qty(\frac{4(j_1+j_2)p_{0}}{1+4(j_1+j_2)p_{0}}).
\end{align}
\end{theorem}

This bound recovers the noiseless case as $p_0\to 0$. 
We consider quantum channels $\mathcal{N}:\mathcal{M}_{d_A}\rightarrow \mathcal{M}_{d_B}$ with dimensions $d_A$ and $d_B$ being a power of $2$, i.e., mapping $j_1=\log(d_A)$ qubits to $j_2=\log(d_B)$ qubits because the concatenated quantum Hamming code and its interfaces act on blocks of qubits. However, any quantum channel can be embedded into a quantum channel that acts on blocks of qubits in this way.

As outlined in Section~\ref{sec-strategy}, we will now use the level conversion argument in Theorem~\ref{thm-level-conversion} to develop a level-$l$ circuit for the channel's coding scheme as shown in Fig.~\ref{fig:coding}. We will use a fault-tolerant implementation in the same quantum error-correcting code $\mathcal{Q}^{(L)}$ with the same concatenation level $L$ for both the channel's encoder and the channel's decoder. This restriction, however, is not imposed by our construction; analogous results can be obtained for fault-tolerant implementations in two distinct codes (as long as they both fulfill some version of a level-conversion theorem) or with different code levels.

Due to the effective scaling of the error in Theorem~\ref{thm-level-conversion}, the analysis of the fault-tolerant channel coding for a noisy channel $\mathcal{N}$ reduces to the analysis of faultless channel coding for an effective channel model with some channel $\mathcal{N}_{\text{eff},L} $, which receives an additional input state $\sigma_\mathcal{S}$, as presented in Fig.~\ref{fig:coding}.
\begin{figure*}[t]
    \centering
    \includegraphics[scale=0.8]{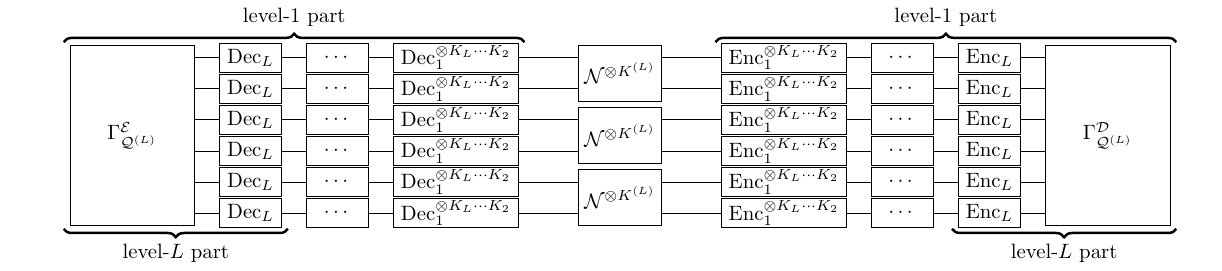}\\
    \raisebox{1.5cm}{\hbox{\Large{$\Rightarrow$}}}
    \includegraphics[scale=0.8]{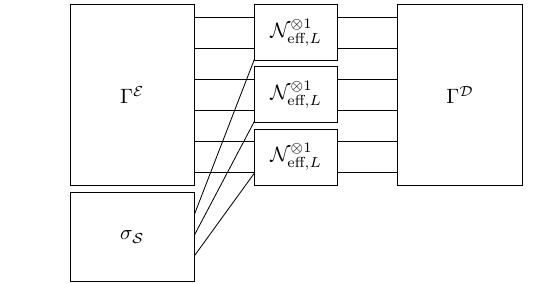}
    \caption{A schematic on the level-$L$ circuit for fault-tolerant channel coding for a quantum channel $\mathcal{N}$ on the top and the corresponding task for an effective channel ${\mathcal{N}}_{\text{eff},L}$ at the bottom.}
    \label{fig:coding}
\end{figure*}
This result for the concatenated quantum Hamming code is an analogue of the corresponding result for the concatenated 7-qubit Steane code in \cite[Lemma~III.4]{CMH20}. However, we emphasize that the effective channel $\mathcal{N}_{\text{eff},L}:\mathcal{M}_2^{\otimes j_1} \otimes \mathcal{M}_2^{\otimes j_1 (A-1)} \rightarrow \mathcal{M}_2^{\otimes j_2}$ with $A=\frac{N^{(L)}}{K^{(L)}}$ differs from its counterpart for the concatenated 7-qubit Steane code, which acts on quantum states in $\mathcal{M}_2^{\otimes j_1} \otimes \mathcal{M}_2^{\otimes j_1 (7^L-1)}$. This difference is a consequence of the respective code's overhead and is very significant in the context of quantum communication (see also Remark~\ref{remark-sigma-overhead}).

\begin{proof}[Proof of Theorem~\ref{thm-cap-lower-bound}] In order to simplify the presentation, first suppose that $j_1=j_2=1$.

\textbf{Mapping the problem to an information theoretic setup with arbitrarily varying perturbations.}
As a first step, we show that a rate $R$ of fault-tolerant entanglement-assisted communication can be achieved if $R$ is an achievable rate of a coding scheme an information-theoretic setup that was previously considered in \cite{BDNW18,Belzig24,CYB25}.

Let $\mc E$ and $\mc D$ be CPTP maps that we specify later. Let $\Gamma^{\mc E}$ and $\Gamma^{\mc D}$ denote the circuits that implement those maps.

We choose a level $L$ such that \begin{align}
    \qty(\qty(\frac{p_0}{\tilde{p}_\mathrm{th}})^{2^{L}}\tilde{p}_\mathrm{th}) (|\Loc( \Gamma^{\mc E} )| + |\Loc( \Gamma^{\mc D} )|)\leq \frac{1}{b}
\end{align}

Let $\Gamma^{\mc E}_{Q^{(L)}}$ and $\Gamma^{\mc D}_{Q^{(L)}}$ denote the implementations of $\Gamma^{\mc E}$ and $\Gamma^{\mc D}$ in $\mc Q^{(L)}$.

Then, by Corollary~\ref{thm-threshold-interfaced}, we have
\begin{align}
&\left\| \left[  \Gamma^{\mc D}_{Q^{(L)}} \circ \qty(\Enc^{(L)})^{\otimes b}\circ \mc N^{\otimes bK^{(L)}} \circ  \qty(\Dec^{(L)})^{\otimes b} \circ  \Gamma^{\mc E}_{Q^{(L)}}\right]_{\mc F (p_0)} - \Gamma^{\mc D} \circ \qty(\mc N_{\textrm{enc}})^{\otimes b}\circ \mc N^{\otimes bK^{(L)}} \circ  \qty(\mc N_{\textrm{dec}})^{\otimes b} \circ ( \Gamma^{\mc E}\otimes \sigma) \right\|_{1\to 1} \notag\\
& \leq \frac{1}{b}
\end{align}
If, for some $R>0$, we have that 
\begin{align} \label{eq-coding-error-norm}
& \left\| \Gamma^{\mc D} \circ \qty(\mc N_{\textrm{enc}})^{\otimes b}\circ \mc N^{\otimes bK^{(L)}} \circ  \qty(\mc N_{\textrm{dec}})^{\otimes b} \circ ( \Gamma^{\mc E}\otimes \sigma) - \id^{\otimes nR } \right\|_{1\to 1}  \to 0,
\end{align}
we also obtain that the following coding error goes to zero: \begin{align}
&\left\|  \left[ \Gamma^{\mc D}_{Q^{(L)}} \circ \qty(\Enc^{(L)})^{\otimes b}\circ \mc N^{\otimes bK^{(L)}} \circ  \qty(\Dec^{(L)})^{\otimes b} \circ  \Gamma^{\mc E}_{Q^{(L)}} \right]_{\mc F (p_0)} - \id^{\otimes nR } \right\|_{1\to 1} \notag\\& \leq \frac{1}{b} +\left\| \Gamma^{\mc D} \circ \qty(\mc N_{\textrm{enc}})^{\otimes b}\circ \mc N^{\otimes bK^{(L)}} \circ  \qty(\mc N_{\textrm{dec}})^{\otimes b} \circ ( \Gamma^{\mc E}\otimes \sigma) - \id^{\otimes nR } \right\|_{1\to 1}\\&
\to 0,
\end{align}
which in turn implies that the fidelity in Eq.~\eqref{eq-coding-error-fidelity} goes to one, and thus we have a valid coding scheme for fault-tolerant entanglement-assisted communication achieving the rate $R$.

Noting that
\begin{align}
    \Gamma^{\mc D} \circ \qty(\mc N_{\textrm{enc}})^{\otimes b}\circ \mc N^{\otimes bK^{(L)}} \circ  \qty(\mc N_{\textrm{dec}})^{\otimes b} \circ ( \Gamma^{\mc E}\otimes \sigma) &= \Gamma^{\mc D} \circ \qty(\mc N_{\textrm{enc}}\circ \mc N^{\otimes K^{(L)}} \circ  \mc N_{\textrm{dec}})^{\otimes b} \circ \qty(\Gamma^{\mc E}\otimes \sigma),
\end{align}
we find a choice of $R, \Gamma^{\mc E}, \Gamma^{\mc D}$ such that Eq.~\eqref{eq-coding-error-norm} does indeed go to zero: by choosing the encoder and decoder for the arbitrarily varying coding scheme for $\mc N_{\textrm{enc}}\circ \mc N^{\otimes K^{(L)}} \circ  \mc N_{\textrm{dec}}$ from \cite{Belzig24}, we obtain that the coding error in Eq.~\eqref{eq-coding-error-norm} vanishes for rates $\tilde{r}$ (in bits per copy of $\mc N_{\textrm{enc}}\circ \mc N^{\otimes K^{(L)}} \circ  \mc N_{\textrm{dec}}$) given by
\begin{align}
    \tilde{r} \leq  C^{\mathrm{ea}}_{\mc J}\qty(\mc N_{\textrm{enc}}\circ \mc N^{\otimes K^{(L)}} \circ  \mc N_{\textrm{dec}})=\sup_{\rho} \inf_{\sigma} I(A':B)_{\id_A\otimes \qty(\mc N_{\textrm{enc}}\circ \mc N^{\otimes K^{(L)}} \circ  \mc N_{\textrm{dec}}) (\rho\otimes \sigma)}, 
\end{align}
and we obtain achievable rates
\begin{align} \label{eq-lowerbound-fqavc}
    C_{\mathcal{F}(p_0)}^\mathrm{ea}(\mc N) \geq \frac{1}{K^{(L)}} C^{\mathrm{ea}}_{\mc J}\qty(\mc N_{\textrm{enc}}\circ \mc N^{\otimes K^{(L)}} \circ  \mc N_{\textrm{dec}}).
\end{align}

\textbf{Effective channel construction.}
We will now make use of the continuity of the capacity $C^{\mathrm{ea}}_{\mc J}$ to connect this expression back to the capacity of the original channel $\mc N$. Note that this is the step that fails in the previous construction from \cite{BCMH24} because the continuity of this capacity is only known to hold for perturbations with a dimension restriction.

In terms of the decomposition of $\mc N_{\textrm{enc}}$ and $\mc N_{\textrm{dec}}$ from Corollary~\ref{thm-threshold-interfaced}, we rewrite the channel as follows: 

\begin{align}
    &\mc N_{\textrm{enc}}\circ \mc N^{\otimes K^{(L)}} \circ  \mc N_{\textrm{dec}}
    \notag\\&=\mc N_{\mathrm{enc},L} \circ \mc N_{\mathrm{enc},L-1}^{\otimes K_L}\circ \cdots \circ \mc N_{\mathrm{enc},2}^{\otimes \frac{K^{(L)}}{K^{(2)}}} \circ \mc N_{\mathrm{enc},1}^{\otimes \frac{K^{(L)}}{K_1}} \circ \mc N^{\otimes K^{(L)}} \circ   \mc N_{\mathrm{dec},1}^{\otimes {K^{(L)}}}\circ  \mc N_{\mathrm{dec},2}^{\otimes \frac{K^{(L)}}{K^{(2)}}}\circ \cdots \circ \mc N_{\mathrm{dec},L-1}^{\otimes K_L }\circ \mc N_{\mathrm{dec},L}
    \\&=\qty(\mc N_{\textrm{enc},L}) \circ \qty(\mc N_{\textrm{enc},L-1})^{\otimes K_L}\circ \cdots \circ \qty(\mc N_{\textrm{enc},2})^{\otimes \frac{K^{(L)}}{K^{(2)}}}  \circ \qty(\mc N_{\textrm{eff},1})^{\otimes K^{(L)}}\circ  \mc N_{\mathrm{dec},2}^{\otimes \frac{K^{(L)}}{K^{(2)}}}\circ \cdots \circ \mc N_{\mathrm{dec},L-1}^{\otimes b K_L }\circ \mc N_{\mathrm{dec},L}
    \\&= \qty(\mc N_{\textrm{enc},L}) \circ \qty(\mc N_{\textrm{enc},L-1})^{\otimes K_L}\circ \cdots \circ \qty(\mc N_{\textrm{eff},l})^{\otimes \frac{K^{(L)}}{K^{(l)}}} \circ \cdots \circ \mc N_{\mathrm{dec},L-1}^{\otimes  K_L }\circ \mc N_{\mathrm{dec},L}
   \\&= \mc N_{\textrm{eff},L},
\end{align}
   where we recursively construct an effective channel $\mc N_{\eff,l}:\mc M_2^{\otimes K^{(l)}}\to\mc M_2^{\otimes K^{(l)}}$ out of the interfaces and the communication channels as
   \begin{align}
       \mc N_{\eff,l}&=\mc N_{\textrm{enc},l-1} \circ \mc N_{\eff,l-1}^{\otimes K^{(l)}} \circ \mc N_{\textrm{dec},l-1}\\& =
    \Big((1-2\tilde{p}_{l-1} )\big( \mc N_{\eff,l-1}^{\otimes K_{l}}\otimes \Tr \big) +2\tilde{p}_{l-1} {\mc W}_{l-1}\Big),
\end{align}
for some channel $\mc W_l$ for each $l$.

For each level $l$ and any state $\sigma$, we thus have \begin{align} \label{eq-continuity-1}
    \left\| \mc N_{\eff,l} (\cdot\otimes \sigma)- \mc N_{\eff,l-1}^{\otimes K_l} \otimes \Tr (\cdot\otimes \sigma) \right\|_{1\to 1} \leq 4 \tilde{p}_{l-1}.
\end{align}

Following \cite{Shirokov17,Belzig24}, for two channels $\mc{N}_1,\mc{N}_2:\mc M_{d_A}\otimes \mc M_{d_S}\rightarrow \mc M_{d_B}$,
we note that if we have,  for all $\sigma \in \mc M_{d_S}$,
\begin{align}
\label{eq:M_sigma}
    \|\mc N_1(\cdot \otimes \sigma ) -\mc N_2 (\cdot \otimes \sigma )\|_{1\to 1}\leq \delta,
\end{align}
then it holds that
\begin{align*}
     C^{\textrm{ea}}_{\mc J}(\mc N_1 )&
     =\inf_{\sigma} C^{\textrm{ea}}(\mc N_1 (\cdot \otimes \sigma) )\\&
     \leq C^{\textrm{ea}}(\mc N_1 (\cdot \otimes \sigma') )
     \\ & \leq C^{\textrm{ea}}(\mc N_2 (\cdot \otimes \sigma') )-g(\delta)\\
     & =  C^{\textrm{ea}}_{\mc J}(\mc N_2 )-g(\delta)
\end{align*}
with $g(\delta)=2\delta\log(d_A) +2 (1+\delta)h_2\qty(\frac{\delta}{1+\delta})$, where we select $\sigma'$ to be the state that achieves the infimum for $C^{\textrm{ea}}_{\mc J}(\mc N_2 )$. The first and last equality follow from \cite{Belzig24} and the second inequality follows from~\cite[Corollary~1]{Shirokov17}.

Using this continuity result, we obtain
\begin{align} \label{eq-continuity-2}
   \qty| C^{\mathrm{ea}}_{\mc J}(\mc N_{\eff,l}) -  C^{\mathrm{ea}}_{\mc J}(\mc N_{\eff,l-1}^{\otimes K_l} \otimes \Tr ) |&=| C^{\mathrm{ea}}_{\mc J}(\mc N_{\eff,l}) -  C^{\mathrm{ea}}_{\mc J}(\mc N_{\eff,l-1}^{\otimes K_l}) |\\&\leq 8 \tilde{p}_{l-1} K^{(l)} +2 (1+4\tilde{p}_{l-1})h_2\qty(\frac{4\tilde{p}_{l-1}}{1+4\tilde{p}_{l-1}}).
\end{align}

Repeated application of Eq.~\eqref{eq-continuity-2} for levels $l=0,1,\cdots,L$ gives 
\begin{align}
     \qty| \frac{1}{K^{(L)}} C^{\mathrm{ea}}_{\mc J}(\mc N_{\eff,L}) -  C^{\textrm{ea}}(\mc N) |& = \qty| \frac{1}{K^{(L)}} C^{\mathrm{ea}}_{\mc J}(\mc N_{\eff,L}) -  C^{\textrm{ea}}(\mc N) + \frac{1}{K^{(L)}}\sum_{l=1}^L (C^{\mathrm{ea}}_{\mc J}(\mc N_{\eff,l-1}^{\otimes K_l}) -C^{\mathrm{ea}}_{\mc J}(\mc N_{\eff,l-1}^{\otimes K_l}) )| \\&
    = \qty|\sum_{l=1}^L  \qty( \frac{1}{K^{(l)}} \qty(C^{\mathrm{ea}}_{\mc J}(\mc N_{\eff,l}) -  C^{\mathrm{ea}}_{\mc J}(\mc N_{\eff,l-1}^{\otimes K_l}) ))|\\&
    \leq\sum_{l=1}^L  \qty( \frac{1}{K^{(l)}} \qty|C^{\mathrm{ea}}_{\mc J}(\mc N_{\eff,l}) -  C^{\mathrm{ea}}_{\mc J}(\mc N_{\eff,l-1}^{\otimes K_l}) |)\\&
    \leq \sum_{l=1}^L \qty( \frac{1}{K^{(l)}}  \qty(8 \tilde{p}_{l-1} K^{(l)} +2 (1+4\tilde{p}_{l-1})h_2\qty(\frac{4\tilde{p}_{l-1}}{1+4\tilde{p}_{l-1}}) ) )\\& \leq \sum_{l=1}^L  \qty(8 \tilde{p}_{l-1} ) + 2 (1+4p_{0})h_2\qty(\frac{4p_{0}}{1+4p_{0}}) 
    \\& \leq \sum_{l=1}^L  \qty(8 \qty(\frac{p_0}{\tilde{p}_\mathrm{th}})^{2^{l-1}}\tilde{p}_\mathrm{th} ) + 2 (1+4p_{0})h_2\qty(\frac{4p_{0}}{1+4p_{0}})
    \\&\leq \frac{8p_0}{1-\frac{p_0}{\tilde{p}_\mathrm{th}}}+2 (1+4p_{0})h_2\qty(\frac{4p_{0}}{1+4p_{0}})
    \\&\leq 16p_0 + 2 (1+4p_{0})h_2\qty(\frac{4p_{0}}{1+4p_{0}}).
\end{align}
The first equality is obtained by adding and subtracting the same terms. The second equality is obtained by rearranging the terms and making use of the fact that $C^{\mathrm{ea}}_{\mc J}$ is single-letter \cite{Belzig24,CYB25}. The first inequality is due to the triangle inequality; the second inequality is obtained by inserting Eq.~\eqref{eq-continuity-2}. Then, in the third inequality, second term is upper bounded because $p_0\leq p_l$ and $L\leq K^{(l)}$. In the fourth inequality, we use Eq.~\eqref{eq:error_rate_universal}, and in the fifth inequality, we use properties of the geometric series. The last inequality holds because we assume that $\tilde{p}_{0}\leq \frac{1}{2}p_{\textrm{th}}$.

\textbf{Arbitrary (finite) channel dimension $j_1,j_2$.}
In this case, we design a coding scheme so that \begin{align}
&\left\| \Gamma^{\mc D} \circ \qty(\mc N_{\textrm{enc}})^{\otimes j_2 b}\circ \mc N^{\otimes bK^{(L)}} \circ  \qty(\mc N_{\textrm{dec}})^{ \otimes j_1 b} \circ ( \Gamma^{\mc E}\otimes \sigma) - \id^{\otimes nR } \right\|_{1\to 1}  \to 0,
\end{align}
i.e. a coding scheme for the communication model from \cite{BDNW18,Belzig24} for the channel \begin{align}
  \mc N_{\textrm{eff},L,j_1,j_2}=  \qty(\mc N_{\textrm{enc}})^{\otimes j_2 }\circ \mc N^{\otimes K^{(L)}} \circ  \qty(\mc N_{\textrm{dec}})^{\otimes j_1}.
\end{align}
We proceed by constructing effective channels (using the union bound)
\begin{equation}
    \mc N_{\eff,l,j_1,j_2}=
    \Big((1-2(j_1+j_2)\tilde{p}_{l-1} )\big( \mc N_{\eff,l-1,j_1,j_2}^{\otimes K_{l}}\otimes \Tr \big) +2(j_1+j_2)\tilde{p}_{l-1} {\mc W}_{l}\Big) \Big( \cdot \otimes \sigma_{l}\Big)
\end{equation}
for some channel $\mc W_l$ and some state $\sigma_l$ for each $l$; then, we apply the continuity bounds from Eq.~\eqref{eq-continuity-1} (with an error of $(j_1+j_2)\tilde{p}_l$ instead of $\tilde{p}_l$) and \eqref{eq-continuity-2} (with an extra factor of $j_2$ from the dimension term) to obtain
\begin{align}
    & \qty| \frac{1}{K^{(L)}} C^{\mathrm{ea}}_{\mc J}(\mc N_{\eff,L,j_1,j_2}) -  C^{\textrm{ea}}(\mc N) |\leq 16(j_1+j_2)p_0j_2 + 2 (1+4(j_1+j_2)p_{0})h_2\qty(\frac{4(j_1+j_2)p_{0}}{1+4(j_1+j_2)p_{0}}).
\end{align}

\end{proof}

\begin{figure}[t]
    \centering
    \includegraphics[width=10cm]{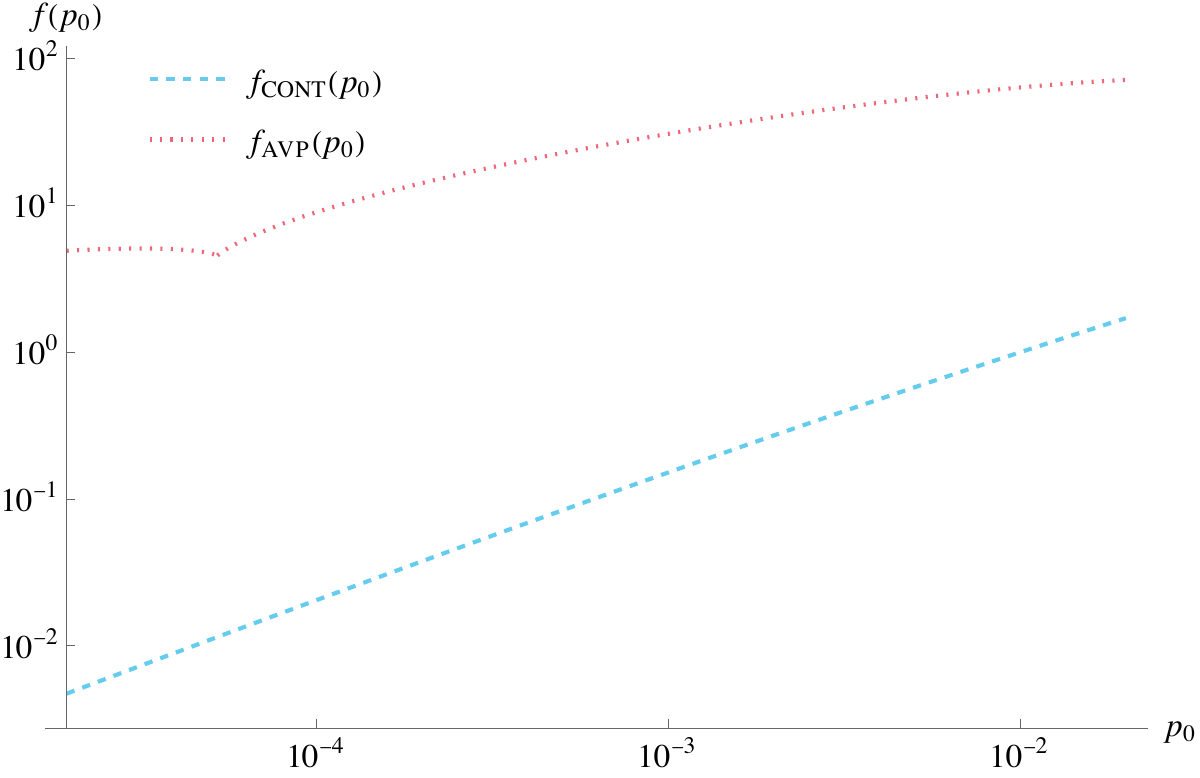}
    \caption{Scaling of the functions by which the noiseless capacity is reduced depending on the local gate error rate $p_0$.}
    \label{fig:capacity-red}
\end{figure}

\begin{remark}[Fault-tolerant channel coding versus a direct application of protocols for fault-tolerant quantum computation] \label{remark-multiplicative-bound}
We briefly discuss an alternative strategy for constructing a fault-tolerant coding scheme, using the example of the entanglement-assisted classical capacity. By simply performing the capacity-achieving coding scheme for more channel copies without leaving the code space, we could hope to lower-bound the fault-tolerant capacity in terms of the faultless capacity and the code's space overhead.

More precisely, let $\mathcal{N}:\mathcal{M}_{d_A}\rightarrow \mathcal{M}_{d_B}$ be a quantum channel. Then, let $\mathcal{E}:\mathbbm{C}^{2^m}\to \mathcal{M}_2^{j_1n}$, $\mathcal{D}:\mathcal{M}_2^{j_2n}\to\mathbbm{C}^{2^m} $ denote an $(n,m,\delta)$-coding scheme for entanglement-assisted communication via $n$ copies of the channel $\mathcal{N}:\mathcal{M}_2^{j_1}\to\mathcal{M}_2^{j_2} $. 
Furthermore, let $\Gamma^{\mathcal{E}}$ and $\Gamma^{\mathcal{D}}$ be the circuits that implement the channels $\mathcal{E}$ and $\mathcal{D}$, and let $\Gamma^{\mathcal{E}}_L:\mathbbm{C}^{2^m}\to \mathcal{M}_2^{j_1n \frac{N^{(L)}}{K^{(L)}}}$ and $\Gamma^{\mathcal{D}}_L:\mathcal{M}_2^{j_1n \frac{N^{(L)}}{K^{(L)}}}\to \mathbbm{C}^{2^m}$ denote the implementations of these circuits in $\mathcal{Q}^{(L)}$.

    Do these circuits form a good coding scheme for $n \frac{N^{(L)}}{K^{(L)}}$ copies of $\mathcal{N}$? If this construction achieves a non-zero communication rate $r$, then $r\frac{K^{(L)}}{N^{(L)}}\leq \frac{r}{36}$ is an achievable rate for fault-tolerant entanglement-assisted communication, where $36$ is an estimate of the inverse of the rate of concatenated quantum Hamming codes in Ref.~\cite{yamasaki2022timeefficient}. However, to the best of our knowledge, it is not known for which channels this is a non-trivial lower bound, and to what extent the fault-tolerant implementation of a fixed coding scheme is a sufficiently good coding scheme for more copies of the same channel for entanglement-assisted communication. Even if non-trivial, the resulting multiplicative lower bound is independent of $p_0$, cannot recover the faultless setting, and is generally less tight than the additive lower bound obtained in Theorem~\ref{thm-cap-lower-bound}.
\end{remark}

\begin{remark}[Difference between our achievable communication rate with the concatenated quantum Hamming code and the previous result with concatenated $7$-qubit Steane code]\label{remark-sigma-overhead}
    In both this result and the previous rate bound \cite{BCMH24}, we find an equivalence between the strategy for constructing a fault-tolerant coding scheme and an effective communication scenario via a different channel acting on an additional input state $\sigma_S$ which may be correlated between the channel uses, and which the sender and the receiver cannot influence. In general, this state may depend on the specific channel, our choice of coding scheme for the channel, our choice of gate set and unitary compilation method, our choice of error-correcting code, level, and compilation procedure, as well as the fault pattern affecting the circuit. This kind of communication model was first studied in~\cite{BDNW18} and has well-understood coding rates only in special cases, for example, when $\sigma_S$ is a product state on its $n$ subsystems. For completely general quantum states $\sigma_S$, the achievable communication rates are not known and, in particular, not known to be non-zero. For this case, which is the relevant case for fault-tolerant communication with the 7-qubit Steane code without making any further assumptions, \cite[Lemma~IV.10]{CMH20} gives a postselection-type bound which they use in order to construct coding schemes and achievable rates for channels of a particular form, namely a convex combination of a channel $\mathcal{N}\otimes\Tr_S:\mathcal{M}_2^{\otimes j_1} \otimes \mathcal{M}_2^{\otimes j_1(7^L-1)}\to \mathcal{M}_2^{\otimes j_2}$, where the extra system is traced out, and a general channel $\mathcal{W}_L:\mathcal{M}_2^{\otimes j_1} \otimes \mathcal{M}_2^{\otimes j_1 (7^L-1)}\to \mathcal{M}_2^{\otimes j_2}$. Similar results are obtained for other error models in \cite[Theorem~63,64]{CFG24}. These results could also be directly applied to fault-tolerant communication with the concatenated quantum Hamming code, which would lead to similar achievable rates to the ones found in~\cite{CMH20,BCMH24} for the concatenated 7-qubit Steane code. However, due to the scaling of the dimension of $\sigma_S$, we recover a special case of the communication model from~\cite{BDNW18}, for which we can simplify and tighten this part of the argument, and for which explicit achievable rates are known~\cite{BDNW18,Belzig24}. Due to this key difference in the dimension scaling of $\sigma_S$, the concatenated quantum Hamming code can be used to construct a fault-tolerant coding scheme with substantially higher achievable rates than the concatenated 7-qubit Steane code.
\end{remark}

The explicit expression for achievable fault-tolerant communication rates is thus given by the noiseless capacity reduced by a function
\begin{align}
    f_\mathrm{CONT}(p_0)=16(j_1+j_2)p_0j_2+(1+4(j_1+j_2)p_0)h_2\qty(\frac{4(j_1+j_2)p_0}{1+4(j_1+j_2)p_0}).
\end{align}
The rates achieved by previous works using the concatenated 7-qubit Steane code for fault-tolerant entanglement-assisted communication given in \cite[Theorem~6.4]{BCMH24} (noting that $f_\mathrm{DIST}(p_0)=0$ for direct comparison) and restated in Eq.~\eqref{eq-bcmh-bound} are additionally reduced by a function $f_\mathrm{AVP}(p_0)=\mathcal{O}(p_0\log p_0)$ which can be shown to dominate the scaling \cite[Section~2.8.2]{Belzig23}, which we illustrate in Fig.~\ref{fig:capacity-red}. By constructing our coding scheme using the concatenated quantum Hamming code instead, we thus get closer to the upper bound (and thereby, the noiseless optimal rate).

\begin{figure}[t]
    \centering
    \includegraphics[width=8cm]{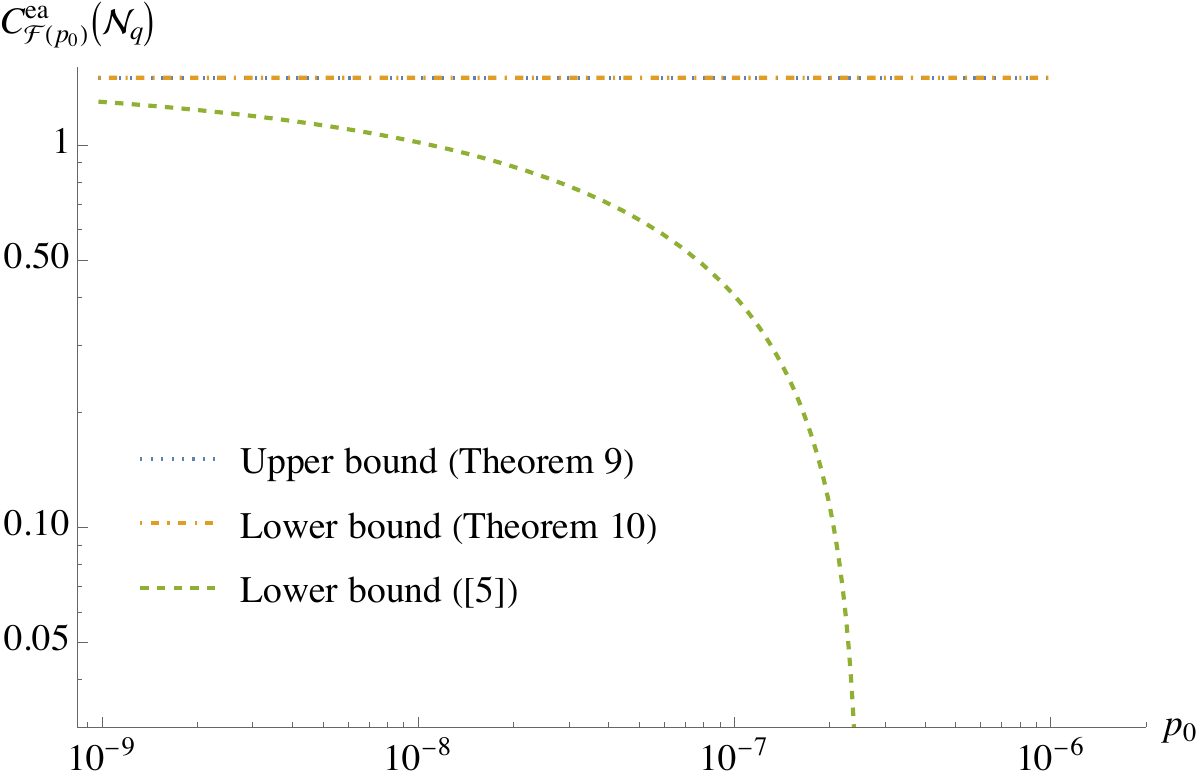}
    \includegraphics[width=8cm]{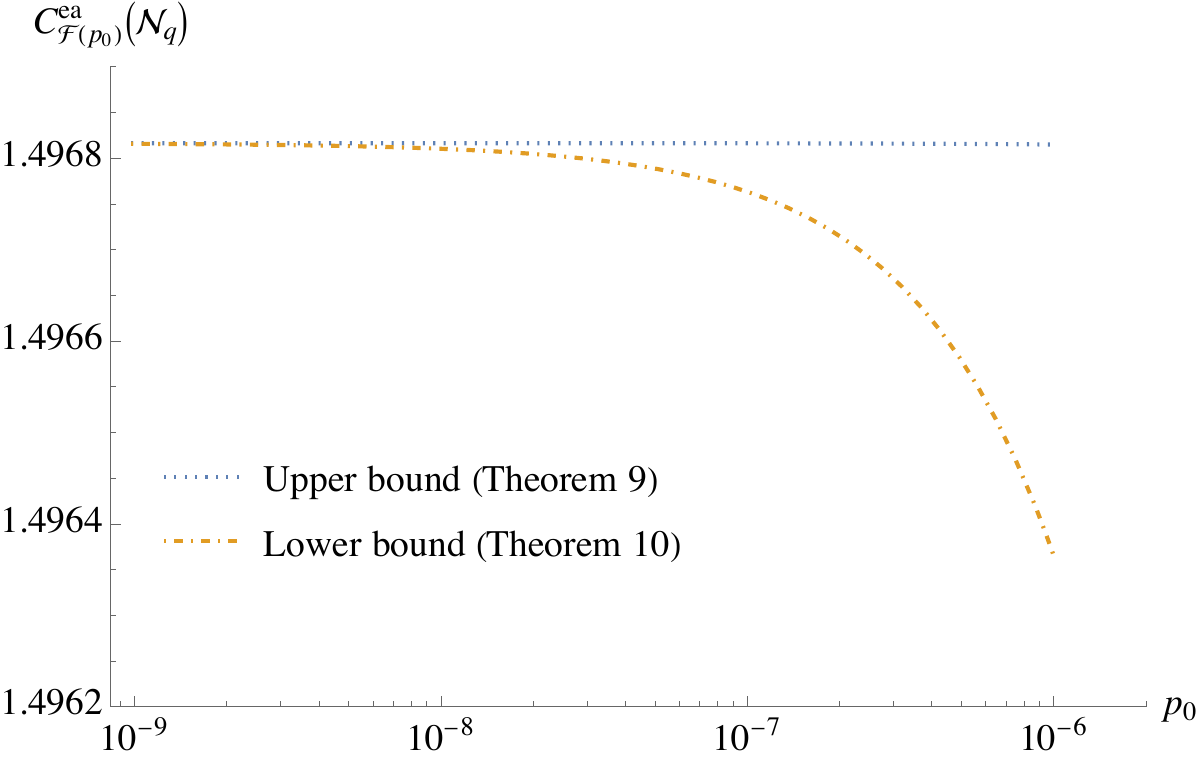}
        \caption{Logarithmic plots of the bounds on the fault-tolerant entanglement-assisted classical capacity of the depolarizing channel $\mathcal{N}_q$ with depolarizing probability $q=0.1$. These plots show the upper bound from Theorem~\ref{thm-cap-upper-bound} in blue. The best lower bound is given by our Theorem~\ref{thm-cap-lower-bound} (orange), which nearly matches the upper bound, as shown in the left plot; in contrast, the lower bound obtained in previous work~\cite{BCMH24} (red) decreases rapidly as the local gate error rate $p_0$ increases. The right plot shows the same bounds with a magnified scale, illustrating that our lower bound still decreases as $p_0 \log p_0$, but at a substantially slower rate. 
        }
    \label{fig:capacity-comp1}
\end{figure}

We illustrate this difference further by considering the example of the qubit depolarizing channel $\mathcal{N}_q(\rho)=(1-q)\rho+q \frac{\mathbbm{1}}{2}$ with depolarizing parameter $q\in[0,1]$. For this channel, an explicit expression of the entanglement-assisted classical capacity is known, which can be found in~\cite{BSST99}. Here, we choose a fixed depolarizing parameter $q=0.1$ and vary the local gate error rate $p_0$ to demonstrate how the upper and lower bounds behave. As expected, both of the lower bounds decrease as $\mc O(p_0\log(p_0))$. However, it is notable that the bound we obtain in Theorem~\ref{thm-cap-lower-bound} is much closer to the upper bound, as can clearly be observed in Fig.~\ref{fig:capacity-comp1}. Even more importantly, it is clear that our lower bound is non-zero for much higher local gate error rates $p_0$ than the previous bound, including for values of $p_0$ that are more likely to be achievable for current state-of-the-art devices. In our example of the qubit depolarizing channel, the lower bound from \cite{BCMH24} becomes zero for the local gate error rate
\begin{align}
    p_0=\min\{2.5\times 10^{-7},\tilde{p}_\mathrm{th}\}.
\end{align}
By contrast, our lower bound remains positive for local gate error rates up to
\begin{align}
    p_0=\min\{0.011,\tilde{p}_\mathrm{th}\}.
\end{align}

\section{Conclusion and open problems}

We have introduced methods for fault-tolerant protocols with the concatenated quantum Hamming code \cite{yamasaki2022timeefficient} when the circuits have quantum input or output. 
Apart from fault-tolerant quantum computation and fault-tolerant channel coding, this construction may also be relevant for dealing with large quantum computations where parts of the device may be subject to higher or unusual errors compared to the rest of the device, for example, in distributed quantum computing setups \cite{Y5,Wehnereaam9288,7562346,PRXQuantum.6.010101}.

Furthermore, we find that fault-tolerance methods for the concatenated quantum Hamming code can considerably reduce the gap between upper and lower bounds on achievable rates for fault-tolerant entanglement-assisted communication, and notably, give non-zero achievable rates for much higher local gate error rates than previous methods. This is not only a consequence of the code's constant overhead, but also of the concatenated structure of the interfaced circuit implementation and the reduced size of syndrome qubit states that our analysis explicitly takes into account. 
It remains an open question whether other code constructions, for example, quantum expander codes \cite{postol2001proposed,7354429,8555154,tamiya2024polylogtimeconstantspaceoverheadfaulttolerantquantum} (which also achieve constant overhead), can implement fault-tolerant coding schemes that achieve similarly high communication rates.

Recently, it has been demonstrated that communication rates similar to those in \cite{CMH20,BCMH24} can be achieved under more general circuit noise models than those considered here. This was shown in \cite{CFG24} for noise models where each gate may be replaced by an arbitrary similar gate,  and in ~\cite{tamiya2024polylogtimeconstantspaceoverheadfaulttolerantquantum} for a circuit-level local stochastic error model with correlated noise. Under such noise models, it is unclear whether the constant overhead of the concatenated quantum Hamming code still allows us to obtain tighter bounds on the achievable rates, which would be an interesting future direction to explore.

\section*{Acknowledgments}

We would like to thank Ashutosh Goswami and Alexander Müller-Hermes for insightful discussions. PB acknowledges funding from the Canada First Research Excellence Fund and the NSERC Alliance Consortia Quantum grant ALLRP 578455-22. HY acknowledges JST PRESTO Grant Number JPMJPR201A, JPMJPR23FC, JSPS KAKENHI Grant Number JP23K19970, and JST CREST Grant Number JPMJCR25I5\@. 

\bibliographystyle{marcotomPB}
\bibliography{citation_bibtex}

\end{document}